\documentclass{article}

\usepackage[latin1]{inputenc}
\usepackage{amsfonts}
\usepackage{amssymb}
\usepackage{times,fullpage}
\usepackage{amsmath,amsthm}
\usepackage{latexsym}

\newcommand{\removeforqip}[1]{}
\newcommand{\removeforfsttcs}[1]{#1}
\newcommand{\COMMENT}[1]{}

\newcommand{\hide}[1]{}

\newcommand{\zo}{\{0,1\}}

\newcommand{\abs}[1]{\mid #1 \mid}

\newcommand{\sgn}{\mathrm{sgn}}

\newcommand{\eps}{\varepsilon}
\renewcommand{\phi}{\varphi}
\newcommand{\ie}{\textit{i.e.},}
\newcommand{\norm}[1]{\parallel #1 \parallel}

\newtheorem{theorem}{Theorem}
\newtheorem{corollary}{Corollary}
\newtheorem{definition}{Definition}
\newtheorem{proposition}{Proposition}
\newtheorem{lemma}{Lemma}
\newtheorem{claim}{Claim}

\newtheorem{defn}{Definition}
\newtheorem{cor}[theorem]{Corollary}

\newcommand{\Ftwo}{\mathbb {GF}_2}

\newcommand{\Real}{\mathbb R}

\newcommand{\enlever}[1]{}
\newcommand{\Prob}{{\rm Prob}}
\newcommand{\NLB}{NLB}
\newcommand{\nlb}{\mbox{non-local box}}
\newcommand{\nlbs}{\mbox{non-local boxes}}
\newcommand{\rk}{{\rm rank_{\mathbb{GF}_2}}}
\newcommand{\rkF}{{\rm rank_{\mathbb{F}}}}

\newcommand{\epsrk}{\eps{-}\rk}
\newcommand{\NL}{NL}
\newcommand{\NLpx}{NL^{||,\oplus}}
\newcommand{\NLp}{NL^{||}}
\newcommand{\NLo}{NL^{\rm ord}}
\newcommand{\NLx}{NL^{\oplus}}
\newcommand{\RNLpx}{NL^{||,\oplus}_\eps}
\newcommand{\RNLp}{NL^{||}_\eps}
\newcommand{\RNLo}{NL^{\rm ord }_\eps}
\newcommand{\RNL}{NL_\eps}

\newcommand{\disc}{Disc}
\newcommand{\galpha}{\gamma_2^\alpha}
\newcommand{\ginf}{\gamma_2^\infty}
\newcommand{\ba}{\mathbf{a}}
\newcommand{\bb}{\mathbf{b}}

\newcommand{\X}{\mathcal{X}}
\newcommand{\Y}{\mathcal{Y}}

\newcommand{\com}{{\rm com}}
\newcommand{\pub}{{\rm pub}}
\newcommand{\priv}{{\rm priv}}
\newcommand{\AND}{AND}
\newcommand{\OT}{OT}
\newcommand{\OTsec}{\widehat{OT}}

\newcommand{\NAND}{N-AND}

\begin{document}

\title{Non-Local Box Complexity and Secure Function Evaluation}
\author{Marc Kaplan\thanks{LRI - Universit\'e Paris-Sud}
\and
Iordanis Kerenidis$^*$
%\thanks{Supported in part by ACI Securit\'e Informatique SI/03 511 and ANR AlgoQP grants of the French Ministry and in part by the European Commission under the Intergrated Project Qubit Applications (QAP) funded by the IST directorate as Contract Number 015848.}\\
%CNRS - LRI\\
%Universit\'e Paris-Sud\\
%jkeren@lri.fr
\and
Sophie Laplante$^*$
% \thanks{LRI - Universit\'e Paris-Sud}
\and
J\'er\'emie Roland\thanks{NEC Laboratories America}
}

\date{}

%\begin{abstract}
%\end{abstract}
%\newpage

\maketitle

\begin{abstract}
A non-local box is an abstract device into which Alice and Bob input bits $x$ and $y$ respectively and receive outputs $a$ and $b$ respectively, where $a,b$ are uniformly distributed and $a \oplus b = x \wedge y$. Such boxes have been central to the study of quantum or generalized non-locality, as well as the simulation of non-signaling distributions. In this paper, we start by studying how many non-local boxes Alice and Bob need in order to compute a Boolean function $f$. We provide tight upper and lower bounds in terms of the communication complexity of the function both in the deterministic and randomized case. We show that non-local box complexity has interesting applications to classical cryptography, in particular to secure function evaluation, and study the question posed by Beimel and Malkin \cite{BM} of how many Oblivious Transfer calls Alice and Bob need in order to securely compute a function $f$. We show that this question is related to the non-local box complexity of the function and conclude by greatly improving their bounds. Finally, another consequence of our results is that traceless two-outcome measurements on maximally entangled states can be simulated with 3 \nlbs, while no finite bound was previously known. 
\end{abstract}

\section{Introduction}

\paragraph{Communication complexity.} Communication complexity is a central model of computation, which was
first defined by Yao in 1979~\cite{yao79}.
\COMMENT{ and has since found numerous applications.}
\removeforfsttcs{It has found applications in
many areas of theoretical computer science including Boolean circuit complexity, time-space tradeoffs, data structures, automata, formula size, etc.}
In this model Alice and Bob receive inputs~$x$ and~$y$ respectively and are allowed to communicate in order to compute a function $f(x,y)$. The goal is to find the minimum amount of communication needed for this task. In different variants of the model, we allow Alice and Bob to err with some probability, and to share common resources in an attempt to enable them to solve their task more efficiently. 

One such resource is shared randomness\removeforqip{, \ie~a common random string that both Alice and Bob know before they receive their inputs}. When Alice and Bob are not allowed any errors, shared randomness does not reduce the communication complexity. On the other hand, when they are allowed to err, a common random string can reduce the amount of communication needed. However, 
Newman's result tells us that shared randomness can be replaced by private randomness at an additional cost logarithmic in the input size~\cite{newman91}. 

Another very powerful shared resource is entanglement. Using teleportation, Alice and Bob can transmit quantum messages by using their entanglement and only classical communication. This model has been proven to be very powerful, in some cases exponentially more efficient than the classical one. Another way to understand the power of entanglement is by looking at the CHSH game~\cite{chsh69}, where Alice and Bob receive uniformly random bits $x$ and $y$ respectively and their goal is to output bits $a$ and $b$ resp. such that $a \oplus b=x\wedge y$ without communicating. It is easy to conclude that even if Alice and Bob share randomness, their optimal strategy will be successful with probability $0.75$ over the inputs. However, if they share entanglement, then there is a strategy that succeeds with probability approximately $0.85$. This game proves that quantum entanglement can enable two parties to create correlations that are impossible to create with classical means. 

\removeforfsttcs{Even though the setting of the previous game is not exactly the same as the model of communication complexity, we can easily transform one to the other. From now on, in our communication complexity model, instead of requiring Bob to output the value of the function $f(x,y)$, we require Alice and Bob to output two bits $a$ and $b$ respectively, such that $a \oplus b = f(x,y)$. We call this ``computing $f$ in parity''.  It is easy to see that the two models are equivalent up to one bit of communication.}

\paragraph{Non-local boxes.} As we said, entanglement enables Alice and Bob to succeed in the CHSH game with probability $0.85$. But what if they shared some resource that would enable them to win the game with probability 1? Starting from such considerations, Popescu and Rohrlich ~\cite{pr94} defined the notion of a \nlb. 
A \nlb\ is an abstract device shared by Alice and Bob.  By one use of a \nlb, we mean that Alice inputs $x$, Bob inputs $y$, Alice gets an output
$a$ and Bob gets $b$ where $a,b$ are uniformly distributed and more importantly $a \oplus b = x \wedge y$. The name \nlb\ is due to the property that one use of a \nlb\ creates correlations between two bits that are maximally non-local (allowing to win the CHSH game with probability one), but still does not allow to communicate, since taken separately, each bit is just an unbiased random coin. As such, a \nlb\ may be considered as a unit of non-locality. We note here an important property of a \nlb, namely that, similar to entanglement, one player can enter an input and receive an output even before the second player has entered an input.  

The importance of the notion of a \nlb\ has become increasingly evident in the last years. Non-local boxes were first introduced to study (quantum or generalized) non-locality. In particular, it was shown than one of the most studied versions of the EPR experiment, where Alice and Bob perform projective measurements on a maximally entangled qubit pair, may be simulated using only one use of a \nlb~\cite{cgmp05}. More generally, it was shown that any non-signaling distribution over Boolean outputs may be exactly simulated with some finite number of \nlbs\ (for finite input size)~\cite{barrettpironio05,jm05}. This was later generalized to any  non-signaling distribution, except that the simulation may not always be performed exactly for non-Boolean outputs~\cite{fw08}. These results rely on the fact that the set of non-signaling distributions is a polytope, so it suffices to simulate the extremal vertices to be able simulate the whole set. In the context of non-locality, another application of \nlbs\ is the study of pseudo-telepathy games~\cite{bm05}.

It is easy to see that one use of a \nlb\ can be simulated with
one bit of communication and shared randomness: Alice outputs a
uniform bit $r$ and sends $x$ to Bob, who outputs $r\oplus x\cdot y$. However, the converse cannot possibly hold, since a \nlb\ cannot be used for communication.

The first question is what happens if we use \nlbs\ as shared resource in the communication complexity model. Van Dam showed \removeforqip{a remarkable result,} that for any Boolean function $f:\zo^n \times \zo^n \rightarrow \zo$, Alice and Bob can use $2^n$ \nlbs\ and no communication at all and at the end output bits $a$ and $b$ such that $a \oplus b=f(x,y)$~\cite{vandam05}. In other words, if \nlbs\ were physically implementable, then all functions would have trivial communication complexity. His results were strengthened by Brassard \textit{et al.} who showed that even if a non-ideal \nlb\ existed, one that solves the CHSH game with probability $0.91$, then still all functions would have trivial communication complexity~\cite{bbl+06}. 
%\comment{add reference to Cleve et al., where the upper bound is the size of the circuit}
Note that in these results, the number of \nlbs\ needed may be exponential in the input size and do not take into account any properties of the function and more precisely its communication complexity without \nlbs. 
It also follows from the work of \cite{barrettpironio05,bbl+06} that for any Boolean function $f$, if $f$ has a circuit with fan-in 2 of size $s$, then there is a deterministic $\nlb$ protocol of complexity $O(s)$, where the bits of the input of $f$ are split arbitrarily among the players. This implies  that exhibiting an explicit function for which the deterministic $\nlb$ complexity is superlinear, would translate into a superlinear circuit lower bound for this function.  This is a notoriously difficult problem, and while a simple counting argument shows that a random function requires exponential size circuits, the best lower bound to date for an explicit function is linear~\cite{LR01,IM02}.

\paragraph{Secure function evaluation.}
Non-local boxes have also been studied in relation to cryptographic primitives such as Oblivious Transfer or Bit Commitment.
Wolf and Wullschleger \cite{ww05}
% studied whether a \nlb\ can be used for secure function evaluation and
showed that Oblivious Transfer is equivalent to a \emph{timed} version of a \nlb\ (up to a factor of 2). 
To maintain the non-signaling property of the \nlb, one can define timed \nlb\ as having  a predefined time limit, and if any of the players have not entered an input by this time, then some fixed input, say 0, is used instead.
%\comment{I think this better corresponds to the simulation, Iordanis, was there a specific reason for your definition?}
%(One might also define it as the \nlb\ waits until both inputs have been entered before sending the outputs.)
Subsequently, Buhrman \textit{et al.} \cite{BCUWW} showed how to construct Bit Commitment and Oblivious Transfer by using \nlbs\ that do not need to be timed but have to be trusted.

In this paper, we are interested in secure function evaluation, which is one of the most fundamental cryptographic tasks. In this model, Alice and Bob want to evaluate some function of their inputs in a way that does not leak any more information than what follows from the output of the function. 
It is known that there exists functions that cannot
be evaluated securely in the information-theoretic setting (\cite{BGW88,CCD88,CK91,K92}).
However, 
all functions can be computed securely in the information theoretic setting if the players
have access to a black box that performs Oblivious Transfer or some other complete function, e.g. the AND function (\cite{GV88,K88}). 

There has been a lot of work trying to identify, in various settings, which functions can be easily evaluated in a secure way, \ie~without any invocation of the black box, and which are hard to evaluate securely, \ie~require at least one invocation of the black box (\cite{CK91,K92,BMM99,K91,KKMO00,K00}).
Moreover, Beaver \cite{B96} showed that there exists a hierarchy of different degrees of hardness for the information-theoretic setting. In other words,
for all $k$, there are functions that can be securely evaluated with $k$ invocations of the AND box but cannot be computed with $k-1$ uses of the black box.

Beimel and Malkin \cite{BM} proposed a quantitative approach to secure function evaluation by studying how many calls to an Oblivious Transfer or other complete black box one needs in order to securely compute a given function $f$ in the honest-but-curious model. For a Boolean function $f:\X\times\Y\to\zo$ and deterministic protocols, they provide a combinatorial characterization of the minimal number of AND calls required, which however does not lead to an efficient algorithm to determine how many ANDs are actually required. They also show that $2^{|\X|}$ ANDs are sufficient for any function. In the randomized case, they provide lower bounds depending on the truth-table of the function which can be at most of the order of $n$. They also state that ``it would be very interesting to try and explore tighter connections with the communication complexity of the functions".

Finally, Naor and Nissim \cite{NN} have given some connections between the communication complexity of a function $f$ and the communication complexity for securely computing $f$. These results, translated into the Beimel-Malkin and our model, only show that the number of ANDs is at most exponential in the communication complexity.

\paragraph{Summary of results.}
In this paper, we provide more evidence on the importance of  non-local boxes by showing how they relate to different models of communication complexity as well as how they can be used as a tool to quantitatively study secure function evaluation.

First, we study how many \nlbs\ are needed in order to distributively compute a Boolean function $f$. 
We define four different variants denoted $NL, \RNL, \NLp, \RNLp$, where the first two are the deterministic and randomized \nlb\ complexity and the latter two are the deterministic and randomized complexity where the \nlbs\ are only used in parallel. 
We provide lower and upper bounds for all these models in terms of communication models and show that in many cases our bounds are tight.

For the deterministic parallel \nlb\ complexity, we show that $\NLp(f)$ is equal to the rank of the function $f$ over $\Ftwo$.
This also implies that it is equivalent to the communication complexity of the function in the following model: Alice and Bob send to a referee one message each and the referee outputs the Inner Product of the two vectors mod 2. 
Moreover we show that $\NLp(f)$ is always greater than the deterministic communication complexity $D(f)$ and less than $2^{D(f)}$. These bounds are optimal as can be seen by looking at the functions of Inner Product and Disjointness. 

In the randomized parallel case, we define a notion of approximate rank over $\Ftwo$ which is equal to $\RNLp(f)$, under the assumption that the output of the protocol is the XOR of the outcomes of the \nlbs. The notion of approximate rank over $\Real$ has been used for communication complexity~\cite{bw99} and gives upper and lower bounds in the randomized model. 

\COMMENT{
{\bf If we show that apprrank over $\Ftwo$ is less than approxrank over $\Real$ then $\RNLp \leq (\gamma_2)^6$}
}

For the deterministic \nlb\ complexity $NL(f)$, we show that it is at least the communication complexity $D(f)$ and, of course, smaller than $\NLp(f)$, which is again a tight bound. In the randomized case, we prove that it is bounded above by the communication complexity $R^{||,MAJ}(f)$ in the following model: Alice and Bob send to a referee one message each and the referee outputs 1 if for the majority of indices, the two messages are equal. This is a natural model of communication complexity that has appeared repeatedly, for example in the simulation of quantum protocols by classical ones and in various upper bounds on simultaneous messages~\cite{KNR,grolmusz, shi05, ls07}. This model is also bounded above in terms of $\ginf$, a quantity which has been used for upper and lower bounds on communication complexity~\cite{ls07}.

In another application of our work, using the recent result of Regev and Toner~\cite{rt07}, we show that traceless two-outcome measurements on maximally entangled states can be simulated with 3 \nlbs. Previously, no finite bound was known for this case.
In order to do this we need to extend our results from Boolean functions to any distribution. 

Then, we look at the consequences of our results in the area of secure function evaluation. The main question we study is how many calls to a secure primitive one needs to make in order to securely evaluate a function~$f$. 
%Such secure primitives include an Oblivious Transfer black box or an AND black box. Since protocols using \nlbs\ are inherently secure due to their non-signaling property, a natural approach would be to show how to reduce \nlbs\ to cryptographic primitives and vice-versa, in order to use our previous results. However, note that such a reduction has been proven to be quite problematic due to the definition of a NLB box as a generalisation of entanglement.
% We will expand on this later on, but we would like to stress that our goal here is different. We would like to provide upper and lower bounds on the number of classical Oblivious Transfer or AND calls that one needs to make using the know-how we obtained from the study of \nlbs.
%We overcome this difficulty by defining a third restricted variant of \nlb\ complexity, where Alice and Bob use their \nlbs\ in the same order. Note that we still use the usual \nlbs\ without any timing properties and hence the protocols remain non-signaling. We show how this model relates to the number of Oblivious Transfer calls and conclude by considerably improving the bounds of Beimel and Malkin. 
%
Specifically, in the honest-but-curious model, we exactly characterize the number of secure AND boxes we need in order to evaluate $f$ by the one-way communication complexity of $f$. Our proof will be reminiscent of our proofs for the \nlb\ complexity. In the malicious model, we upper bound the number of Oblivious Transfer boxes needed by the \nlb\ complexity of $f$, when the \nlbs\ are used in order. This implies strong upper bounds in terms of the communication complexity as well as $\ginf$. For the lower bounds, we show that the communication complexity of $f$ remains a lower bound for \emph{optimal} protocols that securely evaluate~$f$.

Our results show that \nlbs, introduced for the study of quantum correlations or more general non-locality, can provide a novel way of looking at questions about classical communication complexity, secure function evaluation and complexity theory. 

\section{Preliminaries} 

\subsection{Communication Complexity}

Let $f:\X\times\Y\to\zo$ be a bipartite Boolean function. Alice gets an input $x\in \mathcal{X}$ and Bob gets an input $y\in \mathcal{Y}$. We say that Alice and Bob compute $f(x,y)$ in parity if after executing a protocol, Alice outputs a bit $a$ and Bob outputs a bit $b$ such that $a\oplus b=f(x,y)$, where we use $\oplus$ to denote both the logical $XOR$ and addition mod 2. 
This model differs from the standard model, where one of the players
outputs the value of the function, by at most 1 bit.

We use the following notions of communication complexity.  In probabilistic
models, we assume that the players have a common source of randomness.
\begin{itemize}
\item
$D(f)$ and $R_{\eps}(f)$: \COMMENT{denote} deterministic and $\eps$-bounded error communication complexity of $f(x,y)$ in parity.
\item
$D^\rightarrow(f)$ and $R^\rightarrow_{\eps}(f)$: \COMMENT{are the} one-way deterministic and bounded-error communication complexity of $f(x,y)$ in parity.
\item
\COMMENT{Finally, }$D^\parallel(f)$ and $R^\parallel_{\eps}(f)$: \COMMENT{are the} deterministic and bounded-error communication complexities in the model of simultaneous messages, where Alice and Bob each send a message to the referee and the referee outputs
the value of the function $f(x,y)$.
\end{itemize}

For the model of simultaneous messages, we also consider some natural restrictions on how the referee computes
the output from the messages he receives from the players.
We assume the messages sent are of the same length.
Suppose the referee receives bits $\ba= (a_1,\ldots,a_t)$ 
from Alice, and $\bb= (b_1,\ldots,b_t)$ from Bob.  If 
the referee always computes a predefined function $g(\ba,\bb)$,
then we write $D^{\parallel,g}(f)$ or $R^{\parallel,g}_{\eps}(f)$ to be the length of
the message sent by the players (not the sum of these 
lengths, as is done in the standard model).  
In this paper, we will consider two functions, the
inner product modulo 2,
$IP_2(\ba,\bb) = \bigoplus_i(a_i\cdot b_i)$ (where $\cdot$ denotes the multiplication over $\Ftwo$, which corresponds to the logical $AND$)
and the majority function, $MAJ(\ba,\bb) = MAJ(a_1\oplus b_1,\ldots, a_t\oplus b_t)$.

\subsection{Non-local box Complexity}

\begin{defn}[Non-local box]
A $\nlb$ is a device shared by two parties, which on one side takes Boolean input $x$ and immediately produces Boolean output $a$, and on the other side takes Boolean input $y$ and immediately produces Boolean output $b$, according to the following
distribution:
$\mathbf{p}_{NL}(a,b|x,y)=\left\{ 
	\begin{array}{ll} 
			\frac{1}{2}&\text{if } a\oplus b = x\cdot y\\
			0 & \text{otherwise}.
	\end{array}\right.$
\end{defn}

% We use $\oplus$ to denote both the logical $XOR$ and addition mod 2,
% and multiplication (over $\Ftwo$) is written $\cdot$, which
% corresponds to the logical $AND$. 

Let us stress the importance of timing in this definition. Indeed, Alice should receive her output $a$ from the box as soon as she has entered her input $x$, no matter if Bob has already entered his input or not (and vice-versa). This is possible because the input-output distribution is non-signaling, that is, the marginal distribution of Alice's input $a$ does not depend on Bob's input $y$, since $p(a|x,y)=1/2$ for any $a,x,y$. In other words, from Alice's point of view, $a$ is just an unbiased random bit. The reason for this definition is to mimic an EPR experiment, where Alice obtains her measurement outcome as soon as she performs her measurement, independently of whether Bob has performed his measurement or not.

We study a model akin to communication complexity, where Alice and Bob use non-local boxes instead of communication.
In a $\nlb$ protocol, Alice and Bob
wish to compute some function $f:\mathcal{X}\times \mathcal{Y}\rightarrow \{0,1\}$.
Alice gets an input $x\in\mathcal{X}$,
Bob gets an input $y\in \mathcal{Y}$,
and they have to compute $f(x,y)$ in parity.
Recall that it means that at the end
of the protocol, Alice outputs $a\in \{0,1\}$ and Bob
$b\in \{0,1\}$, such that $a\oplus b = f(x,y)$.
For a protocol $P$, we will write $P(x,y)=(a,b)$.
In the course of the protocol, Alice and Bob are allowed shared
randomness and may use $\nlbs$, but they may not communicate.  
Bob is not allowed to see
Alice's inputs to the $\nlbs$, nor does he see the outcome
on Alice's side, and likewise for Alice.

\begin{defn}
For any function $f:\mathcal{X}\times \mathcal{Y}\rightarrow \{0,1\}$,
$\NL(f)$ is the smallest $t$ such that there is a protocol that
computes $f$ in parity exactly, using $t$ $\nlbs$.
\end{defn}

We will label the \nlbs\ with labels from 1 to $t$. (Recall that in general, Alice and Bob are not required to use the $t$ \nlbs\ in the same order.)
We relax the exactness condition and allow the protocol's 
outcome
to be incorrect with constant probability~$\eps$.

\begin{defn}
For any function $f:\mathcal{X}\times \mathcal{Y}\rightarrow \{0,1\}$,
$\RNL(f)$ is the smallest $t$ such that there is a protocol $P$
using $t$ $\nlbs$, with $\Pr[P(x,y)=(a,b) \text{ with } a\oplus b=f(x,y) ] \geq 1-\eps$.
\end{defn}

We will also study two variants of the general model, where the \nlbs\ are used in a restricted manner. First, we assume that 
the \nlbs\ are used in parallel, that is, the
input to any $\nlb$ does not depend on the
outcome of any other. In this model, we denote the complexity
$\NLp$ in the exact case, and
$\RNLp$ in the $\eps$ error case.   

Second, we define the model where both players use the non-local boxes in the same order, that is, the
non-local boxes are labeled from $1$ to $t$ and Alice's input to the \nlb\ with label $i$ does not depend on the outputs from the \nlbs\ labeled from $i+1$ to~$t$ (similarly for Bob).
Note that in the most general case, Alice and Bob may use their $t$ \nlbs\ with labels $1$ through $t$ in whichever order they want. For example, Alice may use the $\nlb$ with label $3$ first, then use the output in order to compute the input for the $\nlb$ with label 1, while Bob might use the $\nlb$ with label 1 first and so forth.
% In this restricted model, we assume that Alice and Bob use their \nlbs\ in the same order, i.e. they both use the $\nlb$ with label 1 first, then the $\nlb$ with label 2 and so on. 
The complexity in this model is denoted $\NLo$ in the exact case, and
$\RNLo$ in the $\eps$ error case. It is clear that this model is more powerful than the parallel model but less powerful than the general \nlb\ complexity. In fact, we will only use this last variant when we talk about secure function evaluation.
\COMMENT{For the relation between communication complexity and \nlbs\ we will only study the general and parallel model.}
Note also that in all these models, the \nlbs\ are still non-signaling and Alice and Bob receive the outputs of the \nlbs\ immediately after they enter their inputs.  

Finally, we consider a restriction where the players always 
output the same predefined function $g$ of the outputs of the $\nlbs$.
Let $(a_1,b_1),\ldots, (a_t,b_t)$
be the outcomes of the $t$ \nlbs\ in some
particular run of a protocol.  
Of particular interest are protocols where Alice
outputs $a=a_1\oplus\cdots \oplus a_t$ and 
Bob outputs $b=b_1\oplus\cdots \oplus b_t$.
The function $g$ is used in a superscript to
denote the  complexity of a function $f$ in this model,
$\NL^g$ in the determinstic case,
and $\RNL^g$ in the $\eps$ error case,
and in particular, $\NLpx$ and $\RNLpx$
when the \nlbs\ are in parallel and $g=\oplus$.

\subsection{Secure Function Evaluation}
We will consider the following cryptographic primitives.
\begin{defn}[Oblivious transfer]
A 2-1 Oblivious Transfer (\OT) is a device which on input bits $p_0,p_1$ for Alice and $q$ for\ Bob, outputs bit $b$ to Bob, such that $b=p_q$.
\end{defn}
\begin{defn}[Secure \AND]
A secure \AND\ is a device which on input bits $p$ for Alice and $q$ for Bob, outputs bit $a$ to Alice, such that $a=p\cdot q$.
\end{defn}
While at first view, these definitions seem similar to the definition of the \nlb, note that the timing properties are different: for the cryptographic primitives, the outputs are produced only after all the inputs have been entered into the device. It is precisely this subtlety that has led to confusion when trying to use \nlbs\ to implement cryptographic primitives, in particular for bit commitment, when timing is particularly important, since a cheating Alice could wait until the reveal phase before committing her bit into the \nlb, without Bob ever realizing it~\cite{BCUWW}. However, we will see that this is not an issue for our results on secure computation.

Let $f:\X\times\Y\to\zo$ be a bipartite Boolean function.
We study the number of cryptographic primitives required to compute $f$.
In all the models we consider, we require perfect privacy. In the honest-but-curious model, perfect privacy means that when a player follows the protocol, he should not learn \emph{more than required} about the other player's input. In the malicious model, this condition must still hold even if the player does not follow the protocol. \emph{Not more than required} means, for models where the function must be computed in parity, that the players should learn nothing about the other's input, while for models where one of the player should output the function, it means that this player should learn nothing more than what he can infer from his input and the value of the function, while the other player should learn nothing.  

Let us note that \AND\ may not be used as a primitive in the malicious model, so we will consider the \OT\ primitive instead.
\COMMENT{
: indeed, a dishonest Alice may input $1$ in all \AND s, and she obtains Bob's input for free, which still allows her to compute the \AND. Therefore, for a dishonest Alice, each \AND\ is just equivalent to a bit of communication (Bob sends his input to Alice), and this does not allow for unconditional secure computation. For this reason, in the malicious model we consider the \OT\ primitive.}
Moreover, in
\COMMENT{the malicious}
this model, it is known that perfect privacy~\cite{dodismicali} cannot be achieved without randomness.
Therefore, in this setting
we do not consider the deterministic model.  Our bounds in the randomized malicious model also hold for the
weaker honest-but-curious model.
\COMMENT{We define}
%%\comment{*** NOTE ABOUT RANDOMIZED VS DETERMINISTIC IN OT VS AND.}
\begin{itemize}
 \item
$\AND(f)$: \COMMENT{to be the} number of secure AND gates required to securely compute $f(x,y)$ (\emph{not} in parity) in the deterministic, honest-but-curious model. We note that we can allow free two-way communication without in fact changing the complexity \cite{BM}.
\item
\COMMENT{Similarly,} $\OT_\eps(f)$: \COMMENT{is the} number of 2-1 Oblivious Transfer calls required to compute $f(x,y)$ in parity with perfect privacy and $\eps$ error over the players' private coins, assisted with (free) two-way communication, in the malicious model.
\end{itemize}

%\item $\AND_s^\com(f)$: same as $\AND_s(f)$, but assisted with (free) two-way communication.
%\item $\NL(f)$: number of \nlbs\ required to compute $f(x,y)$ in parity.
%\item $\NLp(f)$: number of \nlbs\ required to compute $f(x,y)$ in parity, where all \nlbs\ should be used in parallel.
%\item $\NL_s^{\pub,\com}(f)$: number of \nlbs\ required to securely compute $f(x,y)$ in parity, assisted with public coins and (free) two-way communication, in the malicious model.
%\item $\NL_s^{\priv,\com}(f)$: same as $\NL_s^{\pub,\com}(f)$, but with private coins instead of public coins.

\subsection{Complexity Measures}

We will compare \nlb\ complexity to 
traditional models of communication complexity and prove upper
and lower bounds for this new model. Some of these bounds are in terms of the factorization norms of the communication matrix~\cite{ls07} and related measures.  

\begin{definition}
Let $M$ be a real matrix. The $\gamma_2$ norm of $M$ is $\gamma_2(M) = \min_{X^T Y=M} col(X) col(Y)$,
where $col(N)$ is the largest Euclidian norm of a column of $N$.
\end{definition}

It is known that $2\log(\gamma_2(M))$ gives a lower bound on deterministic communication
complexity of $M$, where $M$ is a sign matrix of 
the Boolean function to be computed~\cite{ls07}. 
In order to lower bound the bounded error
communication complexity in the randomized and quantum case,
we consider a ``smoothed'' version of this measure.

\begin{definition}
Let $M$ be a sign matrix and $\alpha \geq 1$. 
$\gamma_2^\alpha (M) = \min \{\gamma_2(N): \forall i,j\ 1 {\leq} M_{i,j} N_{i,j} {\leq} \alpha\}$.
In particular, $\gamma_2^\infty(M)$ is the minimum $\gamma_2$ norm over all matrices $N$
such that $1 \leq M_{i,j} N_{i,j}$.
\end{definition}

The measures $\gamma_2^\alpha$ and $\gamma_2^\infty$ 
give upper and  lower bounds for bounded-error communication complexity~\cite{ls07}: 
$2\log (\galpha(f)/\alpha) \leq R_\eps(f)$  and $R^{||,MAJ}_\epsilon(f) \leq O((\ginf(f))^2)$ (implicit in \cite{ls07}), where
$\alpha=\frac{1}{1-2\eps}$. 

The discrepancy of a sign matrix M over inputs $X\times Y$ with
respect to distribution $\mu$ over the inputs 
is $\disc_\mu(M)=\max_R \sum_{(x,y)\in R} \mu(x,y)M(x,y)$, where
$R$ is taken from all possible rectangles.  
It is known that $\ginf(f) = \Theta (\frac{1}{\disc(f)})$, 
and for any $\alpha$, $\ginf(f) \leq \galpha(f)$~\cite{ls07}.

Finally, for a Boolean function,
the $L_1$ norm is defined as the sum of the absolute values of its Fourier coefficients. 
\removeforfsttcs{\begin{definition}
Let $f: \{-1,1\}^{2n} \rightarrow \{-1,1\}$, and denote by $\alpha_S$ the Fourier coefficients of $f$, that is
$f(x)=\sum_{S \subseteq \{0,1\}^{2n}} \alpha_S \chi_S(x)$ where $\chi_S(x)=\prod_{i\in S} x_i$.
The $L_1$ norm of $f$ is defined by $L_1(f)= \sum_S |\alpha_S|$.
\end{definition}}
We can think of the $2n$ bits of input of the function as equally split between Alice and Bob. Grolmusz uses this notion to upper bound the randomized communication complexity by proving that $R_\epsilon(f) \leq O(L_1^2(f))$ ~\cite{grolmusz}.

\section{Deterministic $\nlb$ complexity}

\subsection{Characterization of $\NLpx$ in terms of rank}
We start by studying a restricted model of $\nlb$ complexity, where the \nlbs\ are used in parallel and at the end of the protocol, Alice and Bob output the parity of the outputs of their \nlbs\ respectively. 
%For a Boolean function $f:\zo^n \times \zo^n \rightarrow \zo$, we denote it by $\NLpx(f)$. 
We will show that the complexity of $f$ in this model is equal to the rank of the communication matrix of $f$ over $\Ftwo$. It is known that this rank is equal to the minimum $m$, such that $f(x,y)$ can be written as $f(x,y) = \bigoplus_{i=1}^{m} a_i(x) \cdot b_i(y)$
(see also~\cite{bw99}).
\removeforfsttcs{

}
This restricted variant of $\nlb$ complexity is exactly the one that appears in van Dam's work~\cite{vandam05}, where he shows that any Boolean function $f$ can be computed by such a protocol of complexity $2^n$. Moreover, we prove that the restriction that the players output the XOR of the outcomes of the \nlbs\ is without loss of generality. 
\begin{theorem}\label{theorem:nlb-rank}
\COMMENT{For any Boolean function $f:\X\times\Y\to\zo$,} $\NLpx(f) = \rk(M_f) = D^{\parallel,IP_2}(f)$. 
\end{theorem}
\begin{proof}
We start by showing that $\NLpx(f) \leq \rk(M_f)$. Let $\rk(M_f)=t$, \ie~$f(x,y)=\bigoplus_{i \in [t]} p_i(x) \cdot q_i(y)$. Then 
we construct a protocol that uses $t$ \nlbs\ in parallel, where Alice and Bob output the parity of the outcomes of the \nlbs\ and for every input $(x,y)$ the output of the protocol is equal to $f(x,y)$. The inputs of Alice and Bob to the $i$-th $\nlb$ are the bits $p_i(x)$ and $q_i(y), i\in [t]$ respectively and let $a_i, b_i$ the outputs of the $\nlb$ such that $a_i \oplus b_i = p_i(x) \cdot q_i(y)$. Alice and Bob output at the end of the protocol the value $(\bigoplus_{i \in [t]} a_i)\oplus(\bigoplus_{i \in [t]} b_i) =  \bigoplus_{i \in [t]} p_i(x) \cdot q_i(y) = f(x,y)$.

Conversely, if there exists a protocol where Alice and Bob use $t$ \nlbs\ in parallel with inputs $p_i(x),q_i(y)$ and outputs $a_i,b_i$, their final output is $(\bigoplus_{i \in [t]} a_i)\oplus(\bigoplus_{i \in [t]} b_i)$ and it always equals $f(x,y)$, then we have 
$f(x,y) = (\bigoplus_{i \in [t]} a_i)\oplus(\bigoplus_{i \in [t]} b_i) = \bigoplus_{i \in [t]} p_i(x) \cdot q_i(y)$ and hence $\rk(M_f) \leq t$.
\removeforfsttcs{

}
From this last argument, we get $D^{\parallel, IP_2}(f) \leq \NLpx(f)$ 
since the players can send $p_i$ and $q_i$ to the referee who computes
the inner product.  For the
converse, if the referee receives $m_A$, $m_B$ from each player and
computes their inner product mod 2, the players can instead input
each bit of the message into a non-local box and output the parity of
the outputs to obtain the same result.
\end{proof}

For the next corollary, we use the fact that %$\rk(M_f) \leq \rkR(M_f)$ and 
$\log(2\rkF(M_f)-1)\leq D(f)+1$ for any field $\mathbb{F}$ 
(see~\cite{KN97}). 
\removeforfsttcs{(The plus one on the right side of the inequality appears because in our model where the value of the function is distributed among the players, the communication complexity can be one bit less than in the standard model.)}

\begin{corollary}\label{cor:NLpx}
\COMMENT{For any Boolean function $f:\X\times\Y\to\zo$,} $\NLpx(f) \leq 2^{D(f)}$. 
\end{corollary}

%In Theorem \ref{simulation}, we provide an explicit protocol for any boolean function $f$ with deterministic communication complexity $D(f)$ that uses $2^{D(f)}-1$ \nlbs.
On the other hand, 
it is easy to see that the one-way communication complexity is a lower bound on the $\nlb$ complexity. 
\COMMENT{Alice can send all her inputs to Bob, and since the non-local box protocol is always correct, they can simulate it, assuming that Alice received only zeros from all non-local boxes.}
\begin{lemma}
\COMMENT{For any Boolean function $f:\X\times\Y\to\zo$, }
$D^\rightarrow(f) \leq  NL(f)$.
\end{lemma}
\begin{proof}
For any deterministic $\nlb$ protocol of complexity $t$, Alice can send her $t$ inputs to the \nlbs\ to Bob, and  since the protocol is always correct, in particular it is correct if both players assume that the output of Alice's \nlbs\ are 0, Alice can output using this assumption.  Bob can then compute his outputs of the \nlbs\ and complete the simulation of the protocol. This shows that the one-way communication complexity is at most $t$. 
\end{proof}
Notice that similarly to the traditional model, this implies an upper bound on the simultaneous messages model when computing in parity as well since for deterministic communication complexity, $D^\parallel(f)\leq D^\rightarrow(f) + D^\leftarrow(f) + 2$.  To see this, it suffices to see that Alice's message plus her output, together with Bob's message plus his output, determine a monochromatic rectangle in the communication matrix.

\subsection{Removing the XOR restriction}
In this section we show that both in the general and in the parallel model of deterministic $\nlb$ complexity, we can assume without loss of generality that the players output the XOR of the outcomes of the $\nlbs$.

\begin{theorem}
\COMMENT{For any Boolean function $f:\{0,1\}^n \times \{0,1\}^n \rightarrow \{0,1\}, \;\;\;\;\;$ } $NL(f) \leq \NLx(f) \leq NL(f)+2$. 
\end{theorem}

\begin{proof}
Let $P$ any protocol that uses $t$ \nlbs\ and at the end Alice outputs $A(x,\ba)$ and Bob outputs $B(y,\bb)$, where $\ba=(a_1,\ldots,a_t)$ and $\bb=(b_1,\ldots,b_t)$ are Alice and Bob's $\nlb$ outputs. Instead of outputting these values they use another two \nlbs\ with inputs $(A(x,\ba)\oplus (\bigoplus_{i \in [t]} a_i) ,1), (1,B(y,\bb)\oplus (\bigoplus_{i \in [t]} b_i))$. Denote by $(a_{t+1}, b_{t+1}), (a_{t+2}, b_{t+2})$ the outputs of the $\nlbs$. 
We have 
\begin{eqnarray*}
a_{t+1} \oplus b_{t+1} = A(x,\ba) \oplus \bigoplus_{i \in [t]} a_i &,&
a_{t+2} \oplus b_{t+2} = B(y,\bb)\oplus \bigoplus_{i \in [t]} b_i 
\end{eqnarray*}
Finally,
\begin{eqnarray*}
(\bigoplus_{i \in [t+2]} a_i) \oplus (\bigoplus_{i \in [t+2]} b_i)
& = & (\bigoplus_{i \in [t]} a_i) \oplus (\bigoplus_{i \in [t]} b_i) \oplus A(x,\ba) \\
&  & \oplus  B(y,\bb) \oplus (\bigoplus_{i \in [t]} a_i) \oplus (\bigoplus_{i \in [t]} b_i) \\
& = & A(x,\ba) \oplus B(y,\bb).
\end{eqnarray*}
\end{proof}

Unlike the general case, showing that in the parallel case we can assume that the players output the XOR of the outputs of the \nlbs\ is not a trivial statement.

\begin{theorem}\label{XOR}
\COMMENT{For any Boolean function $f:\{0,1\}^n \times \{0,1\}^n \rightarrow \{0,1\}$, we have} $\NLp(f) \leq \NLpx(f) \leq \NLp(f)+2$.
\end{theorem}

We proceed by providing two lemmas before proving our theorem.
\begin{lemma}
\label{PQR}
Let $a,b$ the outcomes of a $\nlb$ and $F,G,H$ arbitrary Boolean coefficients that do not depend on $a,b$. 
If for all $a$, ($F\cdot a) \oplus (G \cdot b) \oplus H=0$, then $F=G$.
\end{lemma}

\begin{proof}
Denote by $p,q$ the inputs to the $\nlb$. By setting $a=0$ and $a=1$, we have $(G \cdot p \cdot q) \oplus H=0$ and $F \oplus G \oplus (G \cdot p \cdot q )\oplus H=0$. This implies $F=G$.
\end{proof}

We now fix some notation. Let $f:\{0,1\}^n \times \{0,1\}^n \rightarrow \{0,1\}$ and $P$ a protocol that computes $f$ with zero error and uses $t$ \nlbs\ in parallel.
Let $p_i(x)$, $q_i(y)$ the inputs to the $i$-th $\nlb$ and $a_i$, $b_i$ the corresponding outputs. We also note $\ba=(a_1,\ldots,a_t)$  and $\bb=(b_1,\ldots,b_t)$. Let $A(x,\ba)= \bigoplus_{S\subseteq [t]} A_S(x)\cdot  a_S$ and
$B(y,\bb)= \bigoplus_{S\subseteq [t]} B_S(y)\cdot b_S$ the final outputs of Alice and Bob, where $A_S$ are polynomials in $x$, $B_S$ polynomials in $y$, $a_S=\prod_{i\in S}a_i$ and $b_S=\prod_{i\in S}b_i$. Then, from the correctness of the protocol, we have
\[
\forall (x,y,\ba), f(x,y)=A(x,\ba) \oplus B(y, \bb).
\]

We show that, without loss of generality, we may assume that the inputs to the \nlbs\ satisfy some linear independence condition.

\begin{definition}
A set of bipartite functions $\{f_i(x,y)|i\in T\}$ is linearly independent if $\bigoplus_{i\in T} C_i\cdot f_i(x,y)=\alpha(x)\oplus\beta(y)$, for some $C_i\in\{0,1\}$ and functions $\alpha(x),\beta(y)$, implies $C_i=0\ \forall i\in T$ and $\alpha(x)=\beta(y)$.
\end{definition}

\begin{lemma}
\label{lem:independence}
Let $P$ be a protocol for $f$ using $t$ \nlbs\ in parallel. Then there exists another protocol whose output is always equal to the one of $P$, uses $t'\leq t$ \nlbs\ in parallel and the inputs to the \nlbs\ are such that the set $\{p_i(x) \cdot q_i(y)| i\in [t'] \}$ is linearly independent.
\end{lemma}

\begin{proof}
Suppose that $\bigoplus_{i\in [t]} C_i\cdot p_i(x)\cdot q_i(y)=\alpha(x)\oplus\beta(y)$, with $C_k=1$ for some $k\in [t]$. Then
$p_k(x) \cdot q_k(y)=\alpha(x)\oplus\beta(y)\oplus \bigoplus_{i\in [t]\setminus\{k\}} C_i\cdot p_i(x)\cdot q_i(y)$.
Since $p_k(x) \cdot q_k(y)=a_k \oplus b_k$, Alice and Bob do not need to use the $k$-th $\nlb$ when implementing protocol $P$, it suffices for Alice to set $a_k=\alpha(x)\oplus \bigoplus_{i\in [t]\setminus\{k\}} C_i\cdot a_i$ and for Bob to set $b_k=\beta(y)\oplus \bigoplus_{i\in [t]\setminus\{k\}} C_i\cdot b_i$, which implies a new protocol with $t-1$ $\nlbs$. By repeating this procedure, they can build a protocol using $t'\leq t$ \nlbs\ and such that the whole set $\{p_i(x) \cdot q_i(y)| i\in [t'] \}$ is linearly independent.
\end{proof}

\begin{proof}[Proof of Theorem~\ref{XOR}]
Since by definition $\NLp(f) \leq \NLpx(f)$, it suffices to show that $\NLpx(f) \leq \NLp(f)+2$. Let $\NLp(f)=t$ and let $P$ be a deterministic protocol for $f$, using $t$ $\nlbs$. Let $A(x,\ba)= A_{\emptyset}(x)\oplus \bigoplus_{S\subseteq [t]} A_S(x)\cdot  a_S$
and $B(y,\bb)= B_{\emptyset}(y)\oplus \bigoplus_{S\subseteq [t]} B_S(y)\cdot b_S$ the outputs of Alice and Bob respectively, where the subsets $S$ are non-empty. First, in order to simulate the two local terms $A_{\emptyset}(x)$ and $B_{\emptyset}(y)$, Alice and Bob can use two \nlbs\ with inputs $(A_{\emptyset}(x),1)$ and $(1,B_{\emptyset}(y))$. For the rest of the proof, all the subsets we consider are non-empty. We proceed by proving two claims about the outputs of the protocol.

\begin{claim}\label{claim:XOR-proof-1}
For all $(x,y,\ba)$ and for all $T\subseteq[t],\;\;\;$
$
\bigoplus_{S: T \subseteq S} A_S(x)\cdot a_{S \setminus T} = \bigoplus_{S: T \subseteq S} B_S(y)\cdot b_{S \setminus T}.
$
\end{claim}

\begin{proof}
We prove this claim by induction on the size of $T$.
By definition, the protocol satisfies for all $(x,y,\ba)$,
$f(x,y)=A(x,\ba)\oplus B(y,\bb)$. By factorizing the $k$-th $\nlb$, we get the following expression for every $(x,y,\ba)$:
\[
f(x,y)= \bigoplus_{S: k \notin S} (A_S(x)\cdot a_S \oplus B_S(y)\cdot b_S ) \oplus a_k \cdot(\bigoplus_{S: k \in S} A_S(x)\cdot a_{S \setminus \{k\}}) \oplus b_k\cdot( \bigoplus_{S: k \in S} B_S(y)\cdot b_{S \setminus \{k\}}).
\]
We can now use Lemma~\ref{PQR}, and have that
\[
\bigoplus_{S: k \in S} A_S(x)\cdot a_{S \setminus \{k\}} = \bigoplus_{S: k\in S} B_S(y)\cdot b_{S \setminus \{k\}},
\]
for all $(x,y,\ba)$ and for all $k\in[t]$. Hence, the claim is true for any subset $T$ with $|T|=1$.
Suppose for the induction that it is true for any set of size $n \in [t-1]$ and consider any $T$ such that $|T|=n$. Let $k \notin T$,
\begin{eqnarray*}
\bigoplus_{S: T \subseteq S} A_S(x)\cdot a_{S \setminus T}  &=& \bigoplus_{S: T \subseteq S} B_S(y)\cdot b_{S \setminus T}, \\
\bigoplus_{S: T \subseteq S, k \notin S} A_S(x) a_{S \setminus T} \oplus a_k(\bigoplus_{S: T \cup \{k\} \subseteq S} A_S(x) a_{S \setminus T\cup \{k\}}) &=& \bigoplus_{S: T \subseteq S, k \notin S} B_S(y) b_{S \setminus T} \oplus b_k( \bigoplus_{S: T \cup \{k\} \subseteq S} B_S(y) b_{S \setminus T\cup \{k\}}).
\end{eqnarray*}
Applying Lemma~\ref{PQR} in the previous equation proves the claim for any $T\cup \{k\}$ and hence any set of size $n+1$, which concludes the proof of Claim~\ref{claim:XOR-proof-1}.
\end{proof}

\begin{claim}\label{claim:XOR-proof-2}
For all $(x,y)$ and for all $T\subseteq[t]$, we have:
\begin{eqnarray*}
|T|>1&\Rightarrow& A_T(x) = B_T(y)=0,\\
|T|=1&\Rightarrow& A_T(x) = B_T(y).
\end{eqnarray*}
\end{claim}
\begin{proof}
We prove this claim by downward induction on the size of $T$, starting with $|T|=t$, that is, $T=[t]$. We immediately obtain from Claim~\ref{claim:XOR-proof-1} that $A_{[t]}(x)=B_{[t]}(y)$. As a consequence, these do not depend on $x$ or $y$, and we may define $C_{[t]}=A_{[t]}(x)=B_{[t]}(y)$. Moreover, we can define $A_{S}(x)=B_{S}(y)=0$ for any $S \supseteq [t]$. 
%\comment{ Marc: I don't understand the last sentence here. Do we really need it? Aren't  $A_S$ and $B_S$ only defined for $S \subseteq [t]$.}

Now let $n\geq 2$ and suppose that for any set $S$ of size equal or larger to $n+1$, we have $A_S(x) = B_S(y)=0$, and for any set $S$ of size $n$, we have $A_S(x) = B_S(y)$. From Claim~\ref{claim:XOR-proof-1}, we obtain that for all sets $T$ of size $n-1$,
\[
A_T(x) \oplus \bigoplus_{k\notin T} C_{T\cup\{k\}}\cdot a_k
= B_T(y) \oplus \bigoplus_{k\notin T} C_{T\cup\{k\}}\cdot b_k,
\]
where we have defined $C_S=A_S(x) = B_S(y)$ for all $S$ of size $n$.
Since $a_k\oplus b_k= p_k(x) \cdot q_k(y)$, we have $\bigoplus_{k\notin T} C_{T\cup\{k\}}\cdot p_k(x) \cdot q_k(y)=A_T(x)\oplus B_T(y)$, and by linear independence, we conclude that $C_{T\cup\{k\}}=0$ and $A_T(x)= B_T(y)$. Using the same argument for any $T$ of size $n-1$, we obtain that $A_S(x) = B_S(y)=0$ for all $S$ of size $n$, and $A_T(x)= B_T(y)$ for all $T$ of size $n-1$, which concludes the proof of Claim~\ref{claim:XOR-proof-2}.
\end{proof}

This claim implies that the protocol $P$ outputs the parity of two local terms plus the outcomes of the $\nlbs$, and as a consequence $\NLpx(f)\leq t+2$.
\end{proof}

\subsection{Optimality of our bounds}

We show here that the bounds we proved in the previous section on the parallel and general $\nlb$ complexity ($D(f) \leq \NLp(f) \leq 2^{D(f)}$ and $D(f) \leq \NL(f) \leq \NLp(f)$ respectively) are optimal by giving examples of functions that saturate them.
The first function we consider is the Inner Product function, $IP(x,y)=\oplus_i (x_i \wedge y_i)$, with $x,y \in \zo^n$. For this function we have that $D(IP)=\NL(IP)=\NLp(IP)=n$.

The second function we consider is the Disjointness function, which is equal to 
$DISJ(x,y)= \vee_i (x_i \wedge y_i)$, with $x,y \in \zo^n$. It is well-known that for the communication matrix of the Disjointness function we have $\rk(M_{DISJ})=2^n$ and hence $\NLpx(DISJ)=2^n$. On the other hand, we have $D(DISJ)=n$ and show that $\NL(DISJ)=O(n)$. We describe below a simple protocol for the Disjointness function that follows from \cite{bbl+06} and was pointed out to us by Troy Lee and Falk Unger. 
The Disjointness function also provides an example of an exponential separation between deterministic parallel and general $\nlb$ complexity.

\begin{proposition}
%\paragraph{Protocol for $DISJ$ with deterministic $\nlb$ complexity of $O(n)$\\\\}
$\NL(DISJ) \leq O(n).$
\end{proposition}
\begin{proof}
On input $x=x_1\cdots x_n, y=y_1\cdots y_n$, Alice and Bob use $n$ \nlbs\ with inputs $(x_i,y_i)$ and get outputs $a_i,b_i$ with $a_i \oplus b_i = x_i \cdot y_i$. Then, they can use $2$ \nlbs\ in order to compute the OR of two such distributed bits since
$(a_k \oplus b_k)\vee (a_\ell \oplus b_\ell) = (a_k \vee a_\ell) \oplus (a_k \vee b_\ell) \oplus(b_k \vee a_\ell) \oplus (b_k \vee b_\ell)$. The terms $(a_k \vee a_\ell)$ and $(b_k \vee b_\ell)$ can be locally computed by Alice and Bob respectively and hence they only need to use two \nlbs\ with inputs $(\neg a_k, \neg b_\ell)$ and $(\neg a_\ell, \neg b_k)$ to compute the remaining terms. 
By combining $n$ such distributed OR computations they compute $\vee_i (a_i \oplus b_i)$ and hence output the value of $DISJ(x,y)$ after using $3n$ $\nlbs$. 
\end{proof}

\section{Randomized $\nlb$ complexity}

In this section, we consider protocols that use shared randomness and have success probability at least $2/3$. 
We start by comparing the parallel $\nlb$ complexity to communication complexity. 
Then we exactly characterize $\RNLpx$ in terms of the approximate rank (over $\Ftwo$) of the communication matrix. 

\subsection{Upper and lower bounds for $\NL_\epsilon$}
\begin{theorem}\label{thm:randomized-parallel}
\COMMENT{For any Boolean function $f$,} $R_\epsilon^\rightarrow(f) \leq \RNL(f) \leq \RNLpx(f) \leq 2^{R_\epsilon(f)}$.
\end{theorem}

\begin{proof}
For the first inequality, Alice sends all her inputs to the \nlbs\ to Bob.
They use the shared randomness to simulate the output of Alice's $\nlbs$,
which Alice can use to compute her output, and Bob uses to compute his
outputs to the \nlbs, and compute his output.
%%This bound can be strengthened, by noticing that Alice and Bob can both send their inputs to a referee who can then simulate the protocol. In other words, the parallel $\nlb$ complexity is lower bounded by half the simultaneous message complexity of $f$. 

For the last inequality, let us fix a randomized communication protocol $P$ for $f$
using $t$ bits of communication. 
\COMMENT{
For any input $x$, randomness $r$ and transcript $T$, let $A(T,x,r)$ be 0 if the transcript is not compatible with $x,r$, and equal to Alice's output in $P$ otherwise. Define $B(T,y,r)$ similarly. 
Alice and Bob simulate this protocol by using \nlbs\ in the following way: First, they use their shared randomness as the randomness $r$ of the protocol $P$. Then, having fixed the randomness, they use two \nlbs\ for each possible transcript $T$. In the first one, Alice inputs 1 if and only if the transcript is compatible with her input and the randomness, and Bob inputs $B(T,y,r)$. In the second one, Alice inputs $A(T,x,r)$ and Bob inputs 1 if and only if the transcript is compatible with his input and the randomness. 
Both players finally output the $XOR$ of all the outputs they received from the $\nlbs$.

As the randomness is fixed, there is only one transcript $T_0$ that is compatible at the same time with Alice's and Bob's input. The $XOR$ of the outcomes of the two \nlbs\ that correspond to this transcript is exactly $A(T_0,x,r) \oplus B(T_0,y,r)$ which is equal to the output of the communication protocol $P$.
For all other transcripts, the parity of the outcomes of the corresponding \nlbs\ is always 0. Therefore, we get 
$\RNLpx(f) \leq O(2^{R(f)})$.
}
We can write $P$ as a distribution over deterministic protocols $P_r$
each using at most $t$ bits of communication, and computing some Boolean
function $f_r$.
By Corollary~\ref{cor:NLpx}, $\NLpx(f_r)\leq 2^{t}$.
Taking the same distribution over the \nlb\ protocols
for $f_r$, we get $\RNLpx(f) \leq 2^{t}$ as claimed.
\end{proof}

Note that in fact any $\NL_\epsilon^{\parallel,g}$ protocol can be simulated in the 
simultaneous messages communication model, so in fact $R_\epsilon^\parallel(f)\leq \NL_\epsilon^{\parallel,g}(f)$, for any $g$.

The approximate rank over the reals has been shown to be a useful complexity measure for communication complexity~\cite{bw99}.  For $\nlb$ complexity,
we now define the notion of approximate rank over $\Ftwo$.

\begin{definition} Let $\mathcal P_t$ denote the convex hull of Boolean matrices with rank over $\Ftwo$ at most $t$.  Then for a
\COMMENT{Boolean}
$[0,1]$-valued\footnote{While in this section we are only interested in Boolean matrices, we give the definition for the general case of $[0,1]$-valued matrices as it will be useful in the next section.}
matrix $A$ the approximate rank over $\Ftwo$ is defined by $\epsrk(A) = \min \{t: \exists A' \in \mathcal P_t \text{ with } \norm{A - A'}_\infty \leq \eps\}$.
\end{definition}

%We can now give an equivalent definition which will enables us to relate the approximate rank to the $\nlb$ complexity.
The next proposition gives an alternative definition of the approximate rank
for Boolean matrices.
This definition enables us to relate the approximate rank to the $\nlb$ complexity.

\begin{proposition}
$\epsrk(M_f)$ is the minimum $t$, such that there exists a set of Boolean matrices $A_1,\ldots,A_R$, and a probability distribution over $[R]$ with the following properties: 
\begin{itemize}
\item For every $r \in [R]$, $\rk(A_r) \leq t$,
\item For every $(x,y)$,  $\Prob_{r}[M_f(x,y)=A_r(x,y)] \geq 1-\eps$.
\end{itemize} 
\end{proposition}

\begin{proof}
Suppose that $\epsrk(M_f) =t$, and let $A \in \mathcal P_t$ such that $\norm{{M_f-A}}_\infty \leq \eps$. Denote by $A_1, \ldots, A_d$ the vertices of $\mathcal P_t$. By definition, $A= \sum_i \mu_i A_i$, with $\sum_i \mu_i=1$ and $\forall i, \, \mu_i \geq 0$. 
For any $(x,y)$, picking $A_i(x,y)$ with probability $\mu_i$ has expected value $E_{\mu}(A_i(x,y))=A(x,y)$. It follows that $\Prob_\mu [M_f(x,y) \neq A_i(x,y)] = {\abs{ M_f(x,y) - E_\mu(A_i(x,y)) }}\leq \eps$.
Moreover, for any $i$, $\rk (A_i) \leq t$. This proves that the set $A_1, \ldots, A_d$ and $\mu$ have the desired properties. The proof goes conversely as well.

\end{proof}

\begin{theorem}
For any Boolean function $f$, $\RNLpx(f) = \epsrk(M_f)$
\end{theorem}

\begin{proof}
Fix a randomized protocol for $f$ that uses $t$ \nlbs\ in parallel and is correct with probability at least $1-\eps$. Let $r$ be the string they share in the beginning of the protocol that is drawn according to some probability distribution from a set $R$. Since Alice and Bob output the $XOR$ of the outcomes of the $\nlbs$, the final outcome of the protocol is independent of the inherent randomness of the $\nlbs$, in other words they always compute some function $g_r(x,y)$. Let us the consider the matrices $A_r=M_{g_r}$. We know that $\rk(A_r) \leq t$ for every $r$. Moreover, the correctness of the protocol implies $\Prob_{r}[M_f(x,y)=A_r(x,y)] \geq 1-\eps$. This proves that $\epsrk(M_f) \leq t$. 

Conversely, suppose that $\epsrk(M_f) = t$. Then fix  a set of Boolean matrices $A_1,\ldots,A_R$ such that for every $r \in [R]$, $\rk(A_r) \leq t$ and $\Prob_{r}[M_f(x,y)=A_r(x,y)] \geq 1-\eps$. Consider the following protocol for $f$: Alice and Bob pick at random $r \in R$ and compute the function $g_r(x,y)=A_r(x,y)$. As $\rk (A_r) \leq t$, computing $g_r$ requires at most $t$ \nlbs\ in parallel. It is straightforward to check that this protocol is correct with probability at least $1-\eps$. Hence, $\RNLpx(f) \leq t$.
\end{proof}

%But we don't know how to get rid of the XOR restriction.
In the randomized case, it is easy to get rid of the $XOR$ restriction in general $\nlb$ protocols, since the proof for the deterministic case still goes through. On the other hand, for the parallel case, this appears to be a surprisingly deep question, which remains open.  The main obstacle appears to be related to the inherent randomness of the \nlbs.

Next, we relate the general $\nlb$ complexity to the following model of communication: Alice and Bob send to a referee one message each and the referee outputs 1 if for the majority of indices, the two messages are equal. We denote the communication complexity in this model by $R^{||,MAJ}_\epsilon(f)$. 
This is a natural model of communication complexity that has appeared repeatedly in the simulation of quantum protocols by classical ones, as well as various upper bounds on simultaneous messages~\cite{KNR,grolmusz, shi05, ls07}. 

\begin{theorem}\label{thmrand}
\COMMENT{For any Boolean function $f$, }
$R_\eps^\rightarrow(f) \leq \RNL(f) \leq O(R^{||,MAJ}_\epsilon(f))$.
\end{theorem}

\begin{proof}
\COMMENT{The lower bound proof follows directly from the deterministic case.}
For the lower bound, Alice and Bob can use shared randomness to simulate the output of Alice's \nlbs. Alice then computes her inputs to the \nlbs, and sends them to Bob. From Alice's inputs and outputs to the \nlbs, Bob may compute his inputs and outputs. The players may then compute their outputs to the protocol, which have the same probability distribution as the original protocol. 

For the upper bound, fix a $t$-bit simultaneous protocol for $f$,  where the referee receives two messages $\ba$ and $\bb$ of size $t$ from Alice and Bob and outputs $MAJ(a_1\oplus b_1,\ldots,a_t \oplus b_t)$. It is well-known, by using an addition circuit,  that the majority of $t$ bits can be computed by a circuit of size $O(t)$ with $AND,NOT$ gates. 
Moreover, the distributed $AND$ of two bits can be computed using two $\nlbs$~\cite{bbl+06}.
We conclude that the $\nlb$ complexity of the distributed Majority is $O(t)$ and hence the theorem follows.
\end{proof}

Our theorem implies the following relation between $\nlb$ complexity and factorization norms.

\begin{corollary}\label{cor:factorization-norms}
\COMMENT{For any Boolean function $f$, }
$2\log (\galpha(f)/\alpha) \leq \RNL(f) \leq O((\ginf(f))^2)$, where $\alpha=\frac{1}{1-2\eps}$. 
\end{corollary}

\begin{proof}
It follows from our theorem and the inequalities $2\log (\galpha(f)/\alpha) \leq R_\eps(f)$ (see~\cite{ls07}) and $R^{||,MAJ}_\epsilon(f) \leq O((\ginf(f))^2)$ (also implicitly in \cite{ls07}).
\end{proof}

It is known that $\ginf(f) = \Theta (\frac{1}{Disc(f)})$, and also that for any $\alpha$, $\ginf(f) \leq \galpha(f)$~\cite{ls07}.
Hence, since discrepancy gives a lower bound on the quantum communication complexity with entanglement $Q_\eps^*(f)$~\cite{kremer}, we get the following corollary. 

\begin{corollary}
$NL_\epsilon (f) \leq O(2^{2Q_\eps^*(f)})$.
\end{corollary}

Finally, we can relate the \nlb\ complexity of a function $f$, to the $L_1$ norm of the Fourier coefficients of $f$ by using a result by Grolmusz. Grolmusz showed that for any Boolean function $f$, there exists a randomized public coin protocol that solves $f$ with complexity $O(L_1^2(f))$. This protocol can be easily transformed into a simultaneous messages protocol where the referee outputs the distributed majority of the message bits. Hence,   

\begin{corollary}
$NL_\epsilon(f) \leq O(L_1^2(f))$.
\end{corollary}

Let us make here a last remark about the proof of Theorem \ref{thmrand}.  We started from a Simultaneous Messages protocol where the referee outputs a Majority function and we constructed a \nlb\ protocol with complexity equal to the communication complexity. If we look at this protocol, we can see that Alice and Bob can use their \nlbs\ in the same order. 
This will be useful when we relate \nlbs\ to secure function evaluation.

\begin{corollary}\label{ordered}
\COMMENT{For any Boolean function $f$, }
$R_\eps^\rightarrow(f) \leq \RNL(f) \leq \RNLo(f) \leq O(R^{||,MAJ}_\epsilon(f))$.
\end{corollary}

\subsection{Optimality of our bounds and an efficient parallel protocol for Disjointness}

In the deterministic case, we showed that our bounds are tight and also that the parallel and the general $\nlb$ complexity can be exponentially different. Is the same true for the randomized case?

In fact, the Disjointness and Inner Product functions almost saturate our bound in terms of $\ginf$ for the general randomized $\nlb$ complexity. More precisely, 
for the Disjointness function, we have that $\RNL(DISJ) = \Theta(n)$ (since  $R_{\eps}(DISJ)=\Omega(n)$ and $\NL(DISJ) \leq O(n))$ and 
using discrepancy~\cite[Exercise 3.32]{KN97}, we have $(\ginf(DISJ))^2 = \Theta(n^2)$.
On the other hand, for the Inner Product function we have $\RNL(IP)=\Theta(n)$ but $(\ginf(IP))^2= \Theta(2^n)$.

The case of parallel $\nlb$ complexity is more interesting. We can give a simple parallel protocol for the Disjointness function of complexity $O(n)$, hence showing that the exponential separation does not hold anymore. It is an open question whether or not parallel and general randomized $\nlb$ complexity are polynomially related.

\begin{proposition}
% \paragraph{Protocol for $DISJ$ with parallel $\nlb$ complexity $O(n)$\\\\} 
$\NLp_{1/3}(DISJ) \leq O(n).$
\end{proposition}
\begin{proof}
The idea is to reduce the Disjointness problem to a problem of calculating an Inner Product, which we know how to do with $n$ parallel $\nlbs$. In order to solve the general Disjointness problem with high probability, Alice and Bob proceed as follows: they look at a shared random string $r_1,\ldots,r_n$ and consider the strings $x \wedge r$ and $y \wedge r$ as inputs. In other words, they pick a random subset of their input bits, by picking each index with probability $1/2$. Then they perform an Inner Product calculation on their new inputs by using $n$ \nlbs\ in parallel. Let $a \oplus b = IP(x \wedge r, y \wedge r)$. It is easy to see that if $DISJ(x,y)=0$, then $IP(x \wedge r, y \wedge r)=0$ for all $r$. On the other hand, if $DISJ(x,y)=1$, \ie~if the intersection is non-empty, then  $\Prob_r[ IP(x \wedge r, y \wedge r) = 1]=1/2$, since for a random subset, the probability that the size of the intersection on this subset is odd is exactly the same as the probability that the intersection is even. Hence, %if Alice and Bob output $a,b$ respectively, then 
we have a one-sided error algorithm for Disjointness that is always correct when $DISJ(x,y)=0$ and is correct with probability $1/2$ when $DISJ(x,y)=1$. 

We can get a two-sided error algorithm in the following way:
Alice and Bob simulate the protocol above until they obtain the outputs $a,b$. Then, using their shared randomness, they output $a \oplus b$ with probability $1-p$, and $0\oplus 1$ or $1\oplus 0$ with probability $p/2$. It is easy to see that when $DISJ(x,y)=0$ then the success probability is $1-p$ and when $DISJ(x,y)=1$ the success probability is $p+(1-p)/2 = (1+p)/2$. Taking $p=1/3$ makes the overall success probability of our algorithm~$2/3$.
\end{proof}

\section{Non-local boxes and measurement simulation}

The question of the nature of non-local distributions arising from 
measurements of bipartite quantum states dates back to
Maudlin~\cite{maudlin92}, who used communication complexity
to quantify non-locality.  Perhaps a more natural question
is how many non-local boxes are required to simulate quantum
distributions, since non-local boxes maintain the non-signaling
property of these distributions.  For binary measurements on
maximally entangled qubit pairs
it is known that 1 use of a non-local box suffices~\cite{cgmp05}.

In this section we present another application of our results on $\nlbs$. 
Using the recent breakthrough of Regev and Toner~\cite{rt07}, 
who give a two-bit one-way protocol for 
simulating two-outcome measurements on entangled states
for arbitrary dimensions,
we show that this can be done with 3 $\nlbs$.  
Previously, no finite upper bound was known
for this problem.
%\comment{Cite  a paper of Valerio Scarani? -- I cannot find this paper, just the one for non-maximally entangled qubit pairs using 8 NLB + 4 millionnaire boxes}

Let $\mathbf{p}$ be a distribution over measurement outcomes
$\mathcal{A}\times \mathcal{B}$, conditioned on measurements
$\mathcal{X}\times\mathcal{Y}$.  For measurements on
quantum states, the distribution is non-signaling, that is,
the marginal distributions do not depend on
the other player's measurement: 
\[
\forall a\in\mathcal{A}, b\in \mathcal{B},
x\in \mathcal{X},x'\in \mathcal{X}, y\in\mathcal{Y},y'\in\mathcal{Y}, \;\;\; p(a|x,y) = p(a|x,y') \mbox{ and } p(b|x,y) = p(b|x',y).
\]
Therefore we write the marginals $p(a|x)$ and $p(b|y)$. In this paper we focus on distributions with uniform marginals over $\mathcal{A}=\mathcal{B}= \{0,1\}$. These distributions are in bijection with $[0,1]$-valued matrices.
\begin{defn}[Correlation matrix] 
 Let $\mathbf{p}$ be a distribution with uniform marginals over $\mathcal{A}=\mathcal{B}= \{0,1\}$, conditioned on measurements $\mathcal{X}\times\mathcal{Y}$. The correlation matrix $C_{\mathbf{p}}:\mathcal{X}\times\mathcal{Y}\to [0,1]$ of $\mathbf{p}$ is defined as $C_{\mathbf{p}}(x,y)=\Pr[a\oplus b=1|x,y]$, where $a,b$ are distributed according to $\mathbf{p}$.
\end{defn}
It is not hard to prove that the set of $[0,1]$-valued matrices is the convex hull of the set of Boolean matrices. This implies that the corresponding non-signaling distributions
can be written as convex combinations of distributions of the
following form. For any  $f:\mathcal{X}\times \mathcal{Y}\rightarrow \{0, 1\}$,
we define the associated distribution 
\[\mathbf{p}_{f}(a,b|x,y)=\left\{ 
	\begin{array}{ll} 
			\frac{1}{2}&\text{if } f(x,y)=a\oplus b\\
			0 & \text{otherwise}.
	\end{array}\right. 
\]
In other words, the correlation matrix $C_{\mathbf{p}_f}(x,y)$ is Boolean and coincides with the communication matrix $M_f$.  Observe that
any protocol for $f$ simulates the distribution $\mathbf{p}_f$,
since we may assume without loss of generality that the outcomes
are uniformly distributed (otherwise, Alice and Bob can 
flip their outcomes according to a shared random bit).

Just as for functions, we can define the communication and \nlb\ complexities of a distribution $\mathbf{p}$. When error $\epsilon$ is allowed, we require that the distribution $\mathbf{p}'$ simulated by the protocol be such that $\norm{C_{\mathbf{p}}-C_{\mathbf{p}'}}_\infty\leq\epsilon$.
Since binary distributions with uniform marginals may be represented as convex combinations of distributions arising from functions, we can generalize some results of the previous section to this case:
\begin{theorem}
 For any distribution $\mathbf{p}$ with uniform marginals over $\mathcal{A}=\mathcal{B}= \{0,1\}$, we have
\begin{itemize}
 \item $R_\epsilon^\rightarrow(\mathbf{p}) \leq \RNL(\mathbf{p}) \leq \RNLpx(\mathbf{p}) \leq 2^{R_\epsilon(\mathbf{p})}$,
 \item $\RNLpx(\mathbf{p}) = \epsrk(C_{\mathbf{p}})$.
\end{itemize}
\end{theorem}
\begin{proof}
For the first inequality, the \nlbs\ are simulated by one-way communication the same way as they were for functions: the players use shared randomness to simulate Alice's output to the \nlbs, then Alice computes her inputs to the \nlbs\ according to those outputs, and sends them to Bob (see proof of Theorem~\ref{thmrand} for details).

Let $t=R_\epsilon(\mathbf{p})$, and consider the corresponding randomized communication protocol. Let the shared randomness take value $r$ with probability $p_r$. For any possible value $r$, the protocol will compute a Boolean function $f_r$ such that $D(f_r)\leq t$, and $\norm{C_{\mathbf{p}}-\sum_r p_r M_{f_r}}_\infty\leq \epsilon$. From Corollary~\ref{cor:NLpx}, $ \NLpx(\mathbf{p})\leq 2^t$. Executing the \nlb\ protocol for $f_r$ with probability $p_r$, we obtain a \nlb\ protocol simulating $\mathbf{p}$ with error at most $\epsilon$, so that $\RNLpx(\mathbf{p})\leq 2^t$.

Let $t=\epsrk(C_{\mathbf{p}})$. By definition, there exist Boolean matrices $A_r$ and a probability distribution $p_r$ such that $\rk(A_r)\leq t$ and $\norm{C_{\mathbf{p}}-\sum_r p_r A_r}_\infty\leq\epsilon$. Let $f_r$ be the Boolean function described by communication matrix $A_r$. Theorem~\ref{theorem:nlb-rank} implies that $\NLpx(f_r) \leq t$. Executing the \nlb\ protocol for $f_r$ with probability $p_r$, we obtain a \nlb\ protocol simulating $\mathbf{p}$ with error at most $\epsilon$, so that $\RNLpx(\mathbf{p})\leq t$. The proof goes conversely as well.
\end{proof}

The upper bound on the \nlb\ complexity in terms of the communication complexity may be slightly improved. While this is an insignificant improvement for most applications involving Boolean functions, this becomes relevant when considering low communication complexity distributions, as is the case for some quantum distributions.

\begin{theorem}\label{simulation}
For any distribution $\mathbf{p}$ with uniform marginals over $\mathcal{A}=\mathcal{B}= \{0,1\}$, we have
$\RNLp(\mathbf{p}) \leq 2^{R_\epsilon(\mathbf{p})}-1$.
%  For any non-signaling distribution over binary outputs with uniform marginals, any protocol with $t$ bits of communication can be simulated with $2^{t}-1$ \nlbs\ in parallel.
\end{theorem}
\begin{proof}
We build on an idea presented in~\cite{dlr07} to replace communication by \nlbs.
Let $t=R_\epsilon(\mathbf{p})$, and let us first consider the case of one-way communication protocols.
Denote $m_A(x)$ the message sent by Alice to Bob, $A(x)$ is Alice's output, and $B(m,y)$ is Bob's output when he receives message $m$. %Without loss of generality, we suppose that the four possible messages sent by Alice to Bob are $00,\ 01,\ 10,\ 11$. 
Suppose without loss of generality that one message is exactly the all-zero string $\mathbf{0}$. 
We use one $\nlb$ for each message except $\mathbf{0}$. 
In the $\nlb$ for message $m$, Alice 
inputs $1$ if $m_A(x)=m$ and 0 otherwise. Bob inputs $B(m,y) \oplus B(\mathbf{0},y)$. 
At the end, Alice outputs $\bigoplus a_i \oplus A(x)$ and Bob outputs $\bigoplus b_i \oplus B(\mathbf{0},y)$, where $(a_i,b_i)$ is the output of the $i$-th $\nlb$.
It is easy to check that the output is always $A(x)\oplus B(m_A(x),y)$.

In the two-way communication case, the proof is more involved.
We proceed recursively, at each step removing the last bit of communication.
We handle the different communication scenarios by doubling the
number of protocols at each step.
Throughout this proof, $T^{(k)}$ will be a $k$-bit transcript, 
and $A_i^{(k)}(T^{(k)},x)$ and $B_i^{(k)}(T^{(k)},y)$
will be the 
players' outputs for the $2^{t-k}$ different $k$-bit communication protocols indexed by $i$. The outputs from the different protocols may then be used as inputs to \nlbs, effectively replacing communication by \nlbs.
%\comment{Moreover, for all these protocols, we assume the directions of communication of the successive bits will be the same.  Marc: I'm not sure that we need that assumption. It suffices that for each protocols, each player knows what will be the direction of the messages. This seems implicitely included in the definition of a protocol (otherwise, both player could wait for the other one to send its message forever). Jeremie: I agree with Marc, see the discussion in the proof.}
More specifically, suppose we have a deterministic $t$ bit communication protocol that computes $f(x,y)$ in parity:
$$
f(x,y)=A_0^{(t)}(T^{(t)},x)\oplus B_0^{(t)}(T^{(t)},y).
$$
We prove by downward induction on $k$, from $k=t$ to $k=0$,
that $f(x,y)$ may be written as
\begin{eqnarray}
f(x,y)&=&A_0^{(k)}(T^{(k)},x)\oplus B_0^{(k)}(T^{(k)},y) %\nonumber \\
\oplus \bigoplus_{i=1}^{2^{t-k}-1} A_i^{(k)}(T^{(k)},x)\cdot B_i^{(k)}(T^{(k)},y),\label{eq:induction-nlboxes}
\end{eqnarray}
which shows that $f(x,y)$ can be computed with $k$ bits of communication (to produce the outputs $A_i^{(k)}(T^{(k)},x)$ and $B_i^{(k)}(T^{(k)},y)$), followed by $2^{t-k}-1$ \nlbs\ in parallel.

It will then follow by induction that
\begin{eqnarray*}
f(x,y)&=&A_0^{(0)}(x)\oplus B_0^{(0)}(y)
\oplus \bigoplus_{i=1}^{2^t-1} A_i^{(0)}(x) \cdot B_i^{(0)}(y),
\end{eqnarray*}
so %$m^{NL}(f)\leq 2^t-1$, hence by Theorem~\ref{theorem:nlb-rank}, 
$f(x,y)$ can be computed with $2^t-1$ \nlbs\ in parallel. 
%\comment{Problem to fix: in Theorem~\ref{theorem:nlb-rank}, we prove $\NLpx(f) = \rk(M_f)$ and the rank takes into account local terms, so this is only $m^{NL}(f)$ up to 2.}

Let us consider the $k$-bit protocols with outputs $A_i^{(k)}(T^{(k)},x)$ and $B_i^{(k)}(T^{(k)},y)$ from Eq.~(\ref{eq:induction-nlboxes}) and focus on the $k$-th bit of the transcript $T^{(k)}$. Since both players must agree, depending on the transcript so far $T^{(k-1)}$, whether this bit is communicated by Alice to Bob or vice-versa, we may define a Boolean function $d_k=d_k(T^{(k-1)})$, which gives the direction of this bit, say $d_k(T^{(k-1)})$ is 1 if the bit is communicated by Alice to Bob, and 0 otherwise. Let us now focus on the bit strings $T^{(k-1)}$ such that $d_k(T^{(k-1)})=1$. Since Alice may compute the next bit to be communicated from her input and the $k-1$ first bits of the transcript, we may write it as
$c_i^{(k)} = c_i^{(k)}(T^{(k-1)},x)$, and her output as $A_i^{(k)}(T^{(k-1)},x)$.
As for Bob's output $B_i^{(k)}(T^{(k)},y)$, we use a construction from~\cite{dlr05}
to replace one bit of communication by a non-local box:
\begin{eqnarray*}
B_i^{(k)}(T^{(k)},y)=B_i^{(k)}(T^{(k-1)}0,y)\oplus c_i^{(k)}\cdot
[B_i^{(k)}(T^{(k-1)}0,y)\oplus B_i^{(k)}(T^{(k-1)}1,y)].
\end{eqnarray*}
Therefore, when $d_k(T^{(k-1)})=1$ we may write $f(x,y)$ as
\begin{eqnarray*}
f(x,y)&=&A_0^{(k)}(T^{(k-1)},x)\oplus B_0^{(k)}(T^{(k-1)}0,y)\oplus c_0^{(k)}(T^{(k-1)},x)\cdot
[B_0^{(k)}(T^{(k-1)}0,y)\oplus B_0^{(k)}(T^{(k-1)}1,y)]\\
&&\oplus \bigoplus_{i=1}^{2^{t-k}-1} A_i^{(k)}(T^{(k-1)},x)\cdot B_i^{(k)}(T^{(k-1)}0,y)\\
&&\oplus \bigoplus_{i=1}^{2^{t-k}-1} A_i^{(k)}(T^{(k-1)},x)\cdot c_i^{(k)} \cdot
[B_i^{(k)}(T^{(k-1)}0,y)\oplus B_i^{(k)}(T^{(k-1)}1,y)].
\end{eqnarray*}
For $0\leq i\leq 2^{t-k}-1$, we define the following output functions,
\begin{eqnarray*}
A_i^{(k-1)}(T^{(k-1)},x)&=&A_i^{(k)}(T^{(k-1)},x)\\
B_i^{(k-1)}(T^{(k-1)},y)&=&B_i^{(k)}(T^{(k-1)}0,y)\\
B_{i+2^{t-k}}^{(k-1)}(T^{(k-1)},y)&=&B_i^{(k)}(T^{(k-1)}0,y)\oplus B_i^{(k)}(T^{(k-1)}1,y)\\
A_{2^{t-k}}^{(k-1)}(T^{(k-1)},x)&=&c_0^{(k)}(T^{(k-1)},x),\\
A_{i+2^{t-k}}^{(k-1)}(T^{(k-1)},x)&=&A_i^{(k)}(T^{(k-1)},x)\cdot c_i^{(k)} \quad \text{ (if $i \neq 0$)},
\end{eqnarray*}
when $d_k(T^{(k-1)})=1$, and similar expressions, with $A$ and $B$ swapped, when $d_k(T^{(k-1)})=0$.
Thus, we may write
\begin{eqnarray*}
f(x,y)&=&A_0^{(k-1)}(T^{(k-1)},x)\oplus B_0^{(k-1)}(T^{(k-1)},y)
\oplus \bigoplus_{i=1}^{2^{t-k+1}-1} A_i^{(k-1)}(T^{(k-1)},x)\cdot B_i^{(k-1)}(T^{(k-1)},y).
\end{eqnarray*}
\end{proof}

The simpler proof in the case of one-way protocol can be used to derive an explicit protocol using 3 \nlbs\ to simulate the correlations arising from 2-outcome
measurements made on an entangled bipartite state.  By Tsirelson's theorem~\cite{tsi85}, the problem of simulating these correlations reduces to the following problem.  
\begin{itemize}
\item Alice receives a unit vector $ \vec{x} \in \Real^n $
\item Bob receives a unit vector $ \vec{y} \in \Real^n $
\item Alice outputs $A \in \{-1, 1\}$ ,
Bob outputs $B \in \{-1, 1\}$ such that the correlation
equals the inner product of the two vectors:
$E[AB] =  \vec{x}\cdot \vec{y} $.
\end{itemize}

\begin{cor}
There is a protocol for 
simulating traceless two-outcome measurements on maximally 
entangled states, using 3 \nlbs ~in parallel.
\end{cor}
\begin{proof}
We sketch the protocol of Regev and Toner. 
%By Tsirelson's theorem~\cite{tsi85}, the correlations corresponding to two-outcome measurements on entangled bipartite states are scalar products of unit vectors. 
Assume that the inputs to the problem are two unit vectors $\vec x,\vec y\in \Real^n$.

Alice and Bob share a random dimension 3 subspace of $\Real^n$.
Let $G$ be the matrix of the projection onto this subspace.
Alice and Bob start by applying a transformation
$ \vec{x'}=C( \vec{x}),  \vec{y'}=C( \vec{y})$ (see~\cite{rt07} for details of this transformation $C$), then project their vectors on the random subspace, $ \vec{x''}=G\vec{x'},  \vec{y''}=G\vec{y'}$. Let $\sgn:\Real\mapsto \{-1,1\}$ be the sign function, that is, $\sgn(x)=1$ if $x\geq 0$ and $\sgn(x)=-1$ otherwise.
Alice lets $\alpha_i = \sgn(x''_ i)$ for $i = 0, 1, 2$ and lets
$c_i = \alpha_0\cdot\alpha_i$ for $i = 1, 2$. Alice outputs $A=\alpha_0$
and sends $(c_1 ,c_2)$ to Bob.
Bob outputs $B = \sgn ( \vec{y''}\cdot \vec{z}_{c_1,c_2} ) $, where $\vec{z}_{c_1,c_2}=(1,c_1,c_2)$.

In the protocol with 3 \nlbs, labeled $(1,-1),(-1,1),(-1,-1)$, Alice 
inputs 1 into the box labeled $(c_1,c_2)$ if $(c_1,c_2)\neq (1,1)$, and 0 into
the other boxes.  Bob inputs $(1-\sgn (\vec{y''}\cdot\vec{z}_{m_1,m_2})\cdot(\sgn(\vec{y''}\cdot\vec{z}_{1,1})))/2$ into the box labeled $(m_1,m_2)$.  Let the outputs of the \nlb\ labeled $m$ be $(a_m,b_m)$.  Then Alice outputs $A=\alpha_0\cdot (-1)^{\bigoplus_ma_m}$ and
Bob outputs $B = \sgn(\vec{y''}\cdot\vec{z}_{1,1})\cdot (-1)^{\bigoplus_mb_m}$.

\end{proof}

\section{Secure Function Evaluation}

\subsection{Honest-but-curious model}
%\subsection{Introconversion of cryptographic primitives}

In the honest-but-curious model, it is well known that \OT\ and \AND\ are equivalent (up to a factor of 2).
\begin{claim}
One \AND\ may be simulated by one \OT. One \OT\ may be simulated by two \AND s. These simulations preserve security in the honest-but-curious model.
\end{claim}

As a starting point, we consider the most basic model, namely deterministic secure computation with \AND s in the honest-but-curious model. Beimel and Malkin \cite{BM} have shown that $\AND(f)\leq 2^{|\X|}$. We show that it is characterized by the one-way communication complexity of $f$.

\begin{theorem}\label{claim:one-way-secure-and}
\COMMENT{For any Boolean function $f:\X\times\Y\to\zo$,} $\AND(f)=2^{D^\rightarrow(f)}$.
\end{theorem}

\begin{proof}[Proof]

[$\AND(f)\leq 2^{D^\rightarrow(f)}$]. Let $P$ be a one-way communication protocol for $f$ using $t=D^\rightarrow(f)$ bits of communication, where, on input $x$, Alice sends a message $m(x)\in\zo^t$ to Bob and outputs $A(x)$, while, on input $y$, Bob outputs $B(y,m(x))$.
We now a build a secure protocol for $f$ using $2^t$ secure ANDs. We label the AND gates by a $t$-bit string $i$. Let $m=m(x)$. For the AND gate labeled $i$, Alice inputs $1$ iff $m=i$, while Bob inputs $B(y,i)$. Let $a_i$ be the outputs of the AND gates (received by Alice). Note that $a_m=B(y,m)$, and $a_i=0$ for all $i\neq m$. It then suffices for Alice to output $A(x)\oplus a_m$. The correctness of the protocol is immediate. The privacy for Alice is trivial since Bob does not receive the output of the ANDs, and as a consequence no information from Alice. The privacy for Bob follows from the fact that the only possibly non-zero output that Alice receives from the ANDs is $a_m=B(y,m)$, which she can deduce from $f(x,y)$ and her input $x$.

[$D^\rightarrow(f)\leq \log (\AND(f))$]. Let $P$ be a secure protocol for $f$ using $t=\AND(f)$ AND gates. Beimel and Malkin showed that in the deterministic case, we can assume without loss of generality that there is no communication between Alice and Bob. In the protocol $P$, Alice and Bob input $p_i$ and $q_i$, respectively, in the AND gate labeled $i\in[t]$, and Alice receives the output $a_i=p_i\cdot q_i$. Since Bob does not receive any information, his inputs to the AND gates only depend on his input $y$, that is, $q_i=q_i(y)$. We show that the same holds for Alice.

Let $\ba=(a_1,\cdots,a_t)$ be the vector of outputs from the AND gates. For fixed $x$, since the protocol is deterministic, and Alice should only learn whether $f(x,y)$ is $0$ or $1$, she should only receive two possible vectors, say $\ba^0(x)$ when $f(x,y)=0$ and $\ba^1(x)$ otherwise. Note that if there exists some $x_0\in\X$ such that $f(x_0,y)$ is constant for all $y\in\Y$, say $f(x_0,y)=0$, Alice only receives one possible vector $\ba^0(x_0)$ when $x=x_0$. In that case, we can fix $\ba^1(x_0)$ arbitrarily to any vector different from $\ba^0(x_0)$. Let $m=m(x)$ be the first index such that $a_m^0(x)\neq a_m^1(x)$. For any $i<m(x)$, $a_i^0(x)=a_i^1(x)$, hence Alice knows in advance what outputs she will receive from the first $m(x)-1$ gates. Therefore, Alice does not need these outputs (since she may infer them by herself) and we may assume without loss of generality that she inputs $p_i(x)=0$ in the first $m(x)-1$ gates. For the AND gate number $m$, $a_m^0(x)\neq a_m^1(x)$, so it has to be the case that $p_m(x)=1$ (otherwise $a_m$ is always $0$). From the output of that gate, Alice already knows the value $f(x,y)$ (depending on whether the output is $a_m^0(x)$ or $a_m^1(x)$), so she does not need the outputs of the last $t-m(x)$ AND gates, and we can assume without loss of generality that she just inputs $p_i(x)=0$ for all $i>m(x)$.

To summarize, we can always assume that Alice inputs $p_i(x)=0$ in all AND gates, except for some index $i=m(x)$ where she inputs $1$. For this AND gate, the output will therefore coincide with Bob's input $q_m(y)$. From the definition of $a_m^0(x)$ and $a_m^1(x)$, we then have for this output $q_m(y)=a_m^0(x)$ iff $f(x,y)=0$, that is, in turn, $f(x,y)=q_m(y)\oplus a_m^0(x)$.
We are now ready to build a one-way protocol for $f$. It suffices for Alice to compute the index of the relevant AND gate $m=m(x)$ and to send it to Bob. Then, Bob sets his output to $B(y,m)=q_m(y)$, while Alice sets hers to $A(x)=a_m^0(x)$.
\end{proof}

One can say that this shows that for most functions, randomization is necessary in order to construct efficient protocols even in the honest-but-curious model.

\subsection{Malicious model}

As we said, the \AND\ primitive cannot be used in the malicious model: indeed, a dishonest Alice may input $1$ in all \AND s, and she obtains Bob's input for free, which still allows her to compute the \AND. Therefore, for a dishonest Alice, each \AND\ is just equivalent to a bit of communication (Bob sends his input to Alice), and this does not allow for unconditional secure computation. For this reason, in the malicious model we consider the \OT\ primitive.
Moreover, it is known that in the malicious model, deterministic secure computation is impossible~\cite{dodismicali}, so we consider the case where Alice and Bob may use private coins and the protocol can have $\eps$ error.

Due to their non-signaling property, protocols using \nlbs\ only and no communication, such as those presented in the previous sections, are trivially secure even against malicious players. Indeed, the non-signaling property implies that the view of the protocol by a possibly dishonest player is always independent from the actions of the other player. 
%On the other hand, this strong notion of security also implies that these protocols may never satisfy another notion called \emph{correctness}, which imposes that one player may not corrupt the output of the protocol without being detected.
We show that certain such protocols may be transformed into protocols using \OT s, namely the protocols where Alice and Bob use their \nlbs\ in the same order. At this point, we don't know if this type of protocols are strictly weaker than general \nlb\ protocols. Nevertheless, our upper bounds in terms of communication complexity hold for such protocols as well (Corollary \ref{ordered}) and hence they translate into upper bounds on $\OT_\eps(f)$.

\begin{theorem}
For any \COMMENT{Boolean function $f:\X\times\Y\to\zo$,} $\eps\geq 0$,
$\OT_{\eps}(f) \leq \RNLo(f)$.
\end{theorem}
\begin{proof}
Let us consider a protocol for $f$ using $t$ \nlbs\ in order and no communication. Let us denote $(p_1,\ldots,p_t)$ and $(a_1,\ldots,a_t)$ the inputs and outputs of Alice's \nlbs\ and $(q_1,\ldots,q_t)$ and $(b_1,\ldots,b_t)$ the inputs and outputs of Bob's $\nlbs$. The fact that they use these \nlbs\ in order implies that $p_i$ (and $q_i$) can only depend on the inputs and outputs of the first $(i-1)$ \nlbs\ but not on the remaining ones. Note that still Alice and Bob get their outputs immediately when they enter their inputs.

We now replace each $\nlb$ with an \OT\ starting from the first one, keeping the distribution of the view of the protocol exactly the same.
Alice and Bob know how to pick the inputs to the first $\nlb$ $p_1,q_1$ since they only depend on their inputs $(x,y)$ and the randomness. To replace this $\nlb$, Alice picks a random bit $r_1$ and inputs  $\{r_1,r_1 \oplus p_1\}$ to the \OT\ box; Bob inputs $q_1$ and hence, his output becomes $r_1\oplus p_1\cdot q_1$. Finally, Alice and Bob set the outputs of the simulated $\nlb$ to $a_1=r_1$ and $b_1=r_1\oplus p_1\cdot q_1$. The simulation of the distribution of the outputs of the $\nlb$ is perfect, since $a_1,b_1$ are unbiased random bits and $a_1\oplus b_1=p_1\cdot q_1$.

Alice and Bob continue with the simulation of the remaining \nlbs\ until the end (Alice using a fresh private random bit for each \NLB). Note that for each $\nlb$, Alice and Bob can compute the inputs $p_i,q_i$ from exactly the correct distribution, since they only depend on the previous $(i-1)$ \nlbs\ which have been perfectly simulated. Hence, at the end, we obtain a protocol for $f$ with the same success probability as the original one. Note that this construction works only when the \nlbs\ are used in order. 

It remains to prove that the new protocol with Oblivious Transfer boxes that we constructed is still secure.
Privacy for Bob is immediate since he only interacts with Alice through the \OT s (there is no additional communication), and Alice obtains no output from the \OT s.
Privacy for Alice follows from the fact that the only information that Bob receives from Alice during the protocol is the outputs of the \OT s, and that these output bits are independent from each other and from Alice's input (since Alice uses independent private random bits to generate her \OT\ inputs).
\end{proof}
From the above theorem we can conclude that all the upper bounds that we had for the $\RNLo$ complexity (see Corollaries~\ref{cor:factorization-norms}-\ref{ordered}) translate into upper bounds for $OT_\eps(f)$.

\COMMENT{The construction used to replace a \NLB\ by a \OT\ is due to Wolf and Wullschleger~\cite{ww05}. They use this construction to prove that \OT\ is equivalent to \NLB, but note that this is strictly speaking incorrect due to the different timing properties of \OT\ and \NLB~\cite{BCUWW}.}

The construction used to replace a \NLB\ by a \OT\ is due to Wolf and Wullschleger~\cite{ww05}. In this reference, this construction is used to prove that \OT\ is equivalent to \NLB, but note that this is strictly speaking incorrect due to the different timing properties of \OT\ and \NLB, as pointed out in~\cite{BCUWW}.

We now turn our attention to lower bounds, and for this we need to restrict ourselves to what we call `optimal' secure protocols. An `optimal' secure protocol is one where the function is computed securely in the usual sense, but we also require that for all the \OT\ calls, there is always an input that remains perfectly secure throughout the protocol. Intuitively, since we try to minimize the number of \OT s that we use, it should be the case that these \OT\ calls are really necessary, in the sense that one of the two inputs should always remain secure. If for example both inputs are revealed at some point during the protocol, then one may not use this \OT\ at all, resulting into a more efficient protocol. Even though intuitively our definition seems natural, at this point, we do not know whether this assumption can be done without loss of generality.

Formally, we define optimal secure protocols as follows.
Let us fix some notation. 
Consider a protocol for the secure computation of a function $f$, using
communication and \OT\ boxes.
$A$ denotes the messages from Alice to Bob; $B$ Bob's messages, $S$ and $T$ Alice and Bob's inputs to the \OT\ boxes and $O$ the outputs of the OT boxes. Note that only Bob receives these outputs.
We assume that at every round $i$ of the protocol, $A_i$ is Alice's message, $B_i$ is Bob's message and $S_i=(S^0_i,S^1_i), T_i, O_i$ are the inputs and the output of the $i$-th \OT\ box. (In some rounds we may not have communication, or the communication can proceed in several rounds; these cases can be handled in the same way as in the proof below.) 
$A_{[i]}, B_{[i]}, S_{[i]}, T_{[i]}, O_{[i]}$ is the concatenation of the first $i$ messages of Alice, messages of Bob, inputs and outputs of the $i$-th first \OT\ boxes respectively. 
\begin{definition}\label{def:optimality}
$\OTsec_\eps(f)$ is the number of 2-1 Oblivious Transfer calls required to compute $f(x,y)$ in parity with perfect privacy and $\eps$ error over the players' private coins, assisted with (free) two-way communication, in the malicious model,
subject to the additional conditions that for each $i$,
\begin{eqnarray*}
 \Prob[A_i | A_{[i-1]},B_{[i-1]},S_{[i]},x,r] &=& \Prob[A_i | A_{[i-1]},B_{[i-1]},r],\\
 \Prob[B_i | A_{[i]},B_{[i-1]},T_{[i]},O_{[i]},y,r] &=& \Prob[B_i | A_{[i]},B_{[i-1]},r],
\end{eqnarray*}
where $r$ is the shared random string used in the protocol. %\comment{Is the second condition correct? I was hesitating whether I should also condition on $O_{[i]}$. This could be checked by doing the simulation of the inverted OT boxes and see what condition we need to keep optimality}
\end{definition}

Let us see exactly what our definition says and how it is related to our intuitive definition of `optimal' protocols.
Let us consider the first condition (a similar discussion holds for the second condition, by swapping Alice and Bob's roles).
We claim that the distribution of $A_i$ conditioned on $(A_{[i-1]}, B_{[i-1]},r)$ should be independent of $(S_{[i]},x)$. Imagine that it is not the case. Then there exist two different strings $(S_{[i]},x)$ and $(S_{[i]},x)'$, such that the distribution of $A_i$ conditioned on $(A_{[i-1]}, B_{[i-1]}, r)$ is different depending on whether Alice's inputs are $(S_{[i]},x)$ or $(S_{[i]},x)'$. This is a contradiction to some strong notion of privacy. Bob, knowing $(A_{[i-1]}, B_{[i-1]}, r)$ and receiving $A_i$, will get some information about whether Alice's inputs are $(S_{[i]},x)$ or $(S_{[i]},x)'$. This means that first, if these two strings differ in $x$, then Bob learns information about the input, which cannot happen; and second, if they differ in one of Alice's inputs to some \OT\ box, then the malicious Bob could get information about both Alice's inputs. A malicious Bob can do this by for example picking an \OT\ call at random and input a random bit into the \OT\ box. With non-zero probability, this would be exactly the box for which he can get information about one input bit and with half probability this bit will be different than the one he learned from the \OT\ box. Hence, we believe that our definition captures exactly the notion of an `optimal' protocol where the inputs $x,y$ as well as the inputs to the \OT\ boxes must remain secure throughout the protocol.

\begin{theorem}\label{thm:ot-lower-bound}
For any Boolean function $f$, $\OTsec_{\eps}(f) =  \Omega(R_\eps(f))$
\end{theorem}

\begin{proof}
% The proof of the upper bound follows closely that of the Non-Local Box complexity. Fix a $t$-bit simultaneous protocol for $f$,  where the referee receives two messages $a$ and $b$ of size $t$ from Alice and Bob and outputs $MAJ(a_1\oplus b_1,\ldots,a_t \oplus b_t)$. It is well-known that the majority of $t$ bits can be computed by a circuit of size $O(t)$ with $AND,NOT$ gates. 
% Hence, it suffices to show that the distributed $AND$ of two bits, \ie $AND(a_1\oplus b_1,a_2\oplus b_2)$ can be computed securely using two Oblivious Transfer boxes. This follows exactly the construction with the two Non-Local Boxes. In order to compute the above expression it suffices to compute the expressions $a_1 \wedge b_2$ and $a_2 \wedge b_1$. For each one of them, Alice and Bob use an \OT\ box where Alice inputs $(c_1, a_1+c_1)$ and $(c_2,a_2+c_2)$ and Bob inputs $b_2$ and $b_1$ respectively. It is easy to see that if both parties are honest then at the end of the protocol they have computef $f$ in parity. Moreover the protocol is secure in the malicious model: first, Alice does not receive any information from Bob nor from the \OT\ boxes so she cannot cheat. second, Bob will always receive a random bit from the \OT\ boxes, no matter what he inputs into the box, hence he cannot cheat either. Of course, the \OT\ complexity can be only smaller if we allow communication between Alice and Bob.
% 
% The proof of the lower bound is more involved, because

We want to show that even in the randomized case, communication doesn't help a lot. In other words, we 
want to show that $O(t)$ bits of communication are sufficient, where $t$ is the number of \OT\ boxes the players use.
Recall that we assume that the protocol is optimal in the following sense: privacy is preserved if both the inputs of the players and one of the inputs to the Oblivious Transfer boxes remain secure throughout the protocol. We start with a perfectly secure protocol that uses $t$ OT\ boxes where Bob receives the outputs, and arbitrary two-way communication between Alice and Bob. We show how to obtain a protocol using $O(t)$ bits of communication (and no \OT\ boxes) by proceeding in four steps. 
\begin{enumerate}
 \item First, we show that the optimality conditions imply that we can entirely suppress the communication from Alice to Bob. We defer the proof of this part to the end of the proof. After this first step, we have a perfectly secure protocol using $t$ \OT\ boxes where Bob receives the outputs, no communication from Alice to Bob and arbitrary communication from Bob to Alice.
 \item The second step is to invert the $t$ \OT\ boxes, meaning that we simulate each \OT\ box where Bob receives the output by an \OT\ box where Alice receives the output. It is well-known that this is possible if we add one bit of communication from Alice to Bob
for each \OT\ box~\cite{wwsym}. Therefore, this step requires to reintroduce $t$ bits of communication from Alice to Bob.  Note that the simulation is such that the new protocol is still optimal, that is, it still satisfies the conditions of Definition~\ref{def:optimality}, where Alice and Bob's roles are swapped.
%\comment{Maybe we should provide the details, as Marc wrote: It seems to be a trivial consequence of the simulation, since the in the simulation, Bob's inputs to the TOs are $r$ and $r \oplus c$ where $r$ is a uniform random bit. Certainly Bob's messages cannot depend on $c$, but I think it should be explicitely written somewhere. }
Hence, we now have a perfectly secure protocol using $t$ \OT\ boxes where Alice receives the outputs, $t$ bits of communication from Alice to Bob and arbitrary communication from Bob to Alice.
\item The third step is to suppress the communication from Bob to Alice. For this step, we just need to reuse the analysis of step one. Notice that the situation is similar to step one, since now Alice gets the outputs of the \OT\ boxes. After this, we end up with a perfectly secure protocol using $t$ \OT\ boxes where Alice receives the outputs, $t$ bits of communication from Alice to Bob, and no communication from Bob to Alice.
\item Finally, Alice and Bob can simulate these $t$ \OT\ boxes by communicating $2t$ bits and hence we end up with a communication protocol of complexity $3t$. This concludes the proof of the theorem.
\end{enumerate}

We now show how to perform the first step of the proof.
The goal is to have Alice and Bob use their shared randomness in order to pick Alice's messages without her sending any bit. On the other hand, Bob is going to send the same messages to Alice as before and they will also use the same \OT\ boxes. Bob can fix his private randomness in the beginning of the protocol. Alice is going to start with a uniform distribution on her private randomness and during the protocol she will update this distribution in order to remain consistent with the protocol up to that point.

We now describe the original protocol in more detail:
\begin{itemize}
\item  Alice and Bob pick their private randomness $r_A$ and $r_B$ uniformly at random.
\item For every round $i$ of the protocol
\begin{itemize}
\item Alice and Bob use an \OT\ box with inputs $S_i$ and $T_i$ respectively and Bob receives output $O_i$. Alice's input to the $i$-th \OT\ box is a fixed function of $(A_{[i-1]},B_{[i-1]},S_{[i-1]},x,r,r_A)$ and Bob's input is a fixed function of $(A_{[i-1]},B_{[i-1]},T_{[i-1]},O_{[i-1]},y,r,r_B)$.
\item Alice computes her message $A_i$ as a function of $(A_{[i-1]}, B_{[i-1]},S_{[i]},x,r,r_A)$ and sends it to Bob.
\item Bob computes his message $B_i$ as a function of $(A_{[i]}, B_{[i-1]},T_{[i]},O_{[i]},y,r,r_B)$ and sends it to Alice.
\end{itemize}
\end{itemize}

We look at the distribution induced by this protocol $\Prob[r_A,r_B,A,B,S,T,O | x,y,r]$ and have
\begin{eqnarray*}
\lefteqn{\Prob[r_A,r_B,A,B,S,T,O | x,y,r]}\\
& = & \Prob[r_A] \cdot \Prob[r_B] \cdot \prod_{i\in[t]}  \Prob[S_i | A_{[i-1]}, B_{[i-1]},S_{[i-1]},x,r,r_A]  \\
&& \cdot \Prob[T_i | A_{[i-1]}, B_{[i-1]},T_{[i-1]}, O_{[i-1]},y,r,r_B] \cdot \Prob[O_i | S_{i},T_{i}]\\
&& \cdot \Prob[A_i | A_{[i-1]},B_{[i-1]},S_{[i]},x,r,r_A] 
\cdot   \Prob[B_i | A_{[i]},B_{[i-1]},T_{[i]},O_{[i]},y,r,r_B].
\end{eqnarray*}

We now give a new protocol where Alice and Bob use their shared randomness to simulate Alice's messages. The distribution remains exactly the same as in the original protocol. In order for this to hold, Alice needs to ``update'' her private randomness to retain consistency. Here is the new protocol:

\begin{itemize}
\item Alice and Bob pick their private randomness $r_A$ and $r_B$ uniformly at random.
\item For every round $i$ of the protocol
\begin{itemize}
\item Alice and Bob use an \OT\ box with inputs $S_i,T_i$ respectively and Bob receives output $O_i$. \\
Alice picks her input $S_i$ according to $\Prob[S_i | A_{[i-1]},B_{[i-1]},S_{[i-1]},x,r]$. Bob's input is the same functions of $(A_{[i-1]},B_{[i-1]}, T_{[i-1]},O_{[i-1]},y,r,r_B)$ as in the original protocol.
\item Alice and Bob simulate Alice's message by sampling from the distribution $\Prob[A_i| A_{[i-1]},B_{[i-1]},r]$, i.e. they pick a next message from all messages that are consistent in the original protocol with the shared randomness $r$ and the transcript so far, averaged over the private randomness $r_A$ for Alice. The key point is that due to the optimality condition in Definition~\ref{def:optimality}, Bob also knows the distribution of Alice's message $A_i$ when  averaged over $r_A$, since it does not depend on $(S_{[i]},x)$.
\item Bob computes his message $B_i$ as the same function of 
$( A_{[i]},B_{[i-1]},T_{[i]},O_{[i]},y,r,r_B)$ as in the original protocol and sends it to Alice.
\item Alice ``updates'' her private randomness by picking $r_A$ according to $\Prob[r_A|A_{[i]}, B_{[i]},S_{[i]},x,r]$. 
\end{itemize}
\end{itemize}

We need to show that the distribution corresponding to the above protocol is exactly the same as in the original protocol and also that this a well-defined procedure. We have for the new protocol:

\begin{eqnarray*}
\lefteqn{\Prob[r_A,r_B,A,B,S,T,O | x,y,r]= \Prob[r_A] \cdot \Prob[r_B]}\\
& &  \cdot  \prod_{i\in[t]} \Prob[S_i | A_{[i-1]}, B_{[i-1]},S_{[i-1]},x,r] \cdot\Prob[T_i | A_{[i-1]}, B_{[i-1]},T_{[i-1]},O_{[i-1]},y,r,r_B] \cdot \Prob[O_i | S_{i},T_{i}]\\  
&  & \cdot \Prob[A_i | A_{[i-1]},B_{[i-1]},r] \cdot \Prob[B_i | A_{[i]},B_{[i-1]},T_{[i]},O_{[i]},y,r,r_B] \cdot 
\frac{\Prob[r_A|A_{[i]}, B_{[i]},S_{[i]},x,r]}{\Prob[r_A|A_{[i-1]},B_{[i-1]},S_{[i-1]},x,r]}
\end{eqnarray*}

Note that as it should be the distribution of $r_A$ after the $\ell$-th round is exactly 
\[ \Prob[r_A] \cdot \prod_{i \in [\ell]} \frac{\Prob[r_A|A_{[i]}, B_{[i]},S_{[i]},x,r]}{\Prob[r_A|A_{[i-1]},B_{[i-1]},S_{[i-1]},x,r]} = \Prob[r_A|A_{[\ell]}, B_{[\ell]},S_{[\ell]},x,r]
\]

We now show that the distributions which correspond to the two protocols are the same. It is easy to see that we need to prove the following fact
\begin{eqnarray*}
\lefteqn{\Prob[S_i | A_{[i-1]}, B_{[i-1]},S_{[i-1]},x,r,r_A]
\cdot \Prob[A_i | A_{[i-1]},B_{[i-1]},S_{[i]},x,r,r_A] =}\\
&& \Prob[S_i | A_{[i-1]}, B_{[i-1]},S_{[i-1]},x,r] \cdot \Prob[A_i | A_{[i-1]},B_{[i-1]},r] \cdot \frac{\Prob[r_A|A_{[i]}, B_{[i]},S_{[i]},x,r]}{\Prob[r_A|A_{[i-1]},B_{[i-1]},S_{[i-1]},x,r]}
\end{eqnarray*}

We have

\begin{eqnarray*}
\lefteqn{\Prob[S_i | A_{[i-1]}, B_{[i-1]},S_{[i-1]},x,r,r_A]
\cdot \Prob[A_i | A_{[i-1]},B_{[i-1]},S_{[i]},x,r,r_A]}\\
& = & \Prob[S_i, A_i | A_{[i-1]}, B_{[i-1]},S_{[i-1]},x,r,r_A] \; = \;
\frac{\Prob[S_i, A_i,r_A | A_{[i-1]}, B_{[i-1]},S_{[i-1]},x,r]}
{\Prob[r_A|A_{[i-1]},B_{[i-1]},S_{[i-1]},x,r]} \\
& = & \Prob[S_i, A_i | A_{[i-1]}, B_{[i-1]},S_{[i-1]},x,r] \cdot
\frac{\Prob[r_A|A_{[i]}, B_{[i-1]},S_{[i]},x,r]}{\Prob[r_A|A_{[i-1]},B_{[i-1]},S_{[i-1]},x,r]}\\
& = & \Prob[S_i | A_{[i-1]}, B_{[i-1]},S_{[i-1]},x,r] \cdot \Prob[A_i | A_{[i-1]},B_{[i-1]},r] \cdot \frac{\Prob[r_A|A_{[i]}, B_{[i]},S_{[i]},x,r]}{\Prob[r_A|A_{[i-1]},B_{[i-1]},S_{[i-1]},x,r]}
\end{eqnarray*}

For the last equation we used first that $\Prob[r_A|A_{[i]}, B_{[i-1]},S_{[i]},x,r] = \Prob[r_A|A_{[i]}, B_{[i]},S_{[i]},x,r]$ since $B_i$ is independent of $r_A$ (for fixed $A_{[i]},S_{[i]}$) and more importantly that 
we have
$\Prob[A_i | A_{[i-1]},B_{[i-1]},S_{[i]},x,r] = \Prob[A_i | A_{[i-1]},B_{[i-1]},r]$. The last equality comes from the privacy of the protocol and the fact that it is optimal.

Now, let us make sure that all these probabilities are non-zero. This follows from the privacy of the protocol. Let's say that after the $(i{-}1)$-th round Alice has been able to update her private randomness to a consistent $r_A$. In the next round, Alice picks input $S_i$ for the \OT\ box from the distribution of all inputs $S_i$ that are consistent for some $r_A$ and Alice and Bob pick a message $A_i$ as Alice's next message. The distribution from which Alice and Bob picked $A_i$ is from all messages for which for shared randomness $r$, there exists an input $x$ and inputs $S_{[i]}$ such that there exists a string $r_A$ so that $A_i$ is consistent with $(x,r,r_A,A_{[i-1]},S_{[i]})$. 
From privacy, if there exists an $x$ and inputs $S_{[i]}$ for which the transcript $A_{[i]}$ is consistent for some $r_A$, then it has to be consistent for all inputs $x$ and all inputs $S_{[i]}$ and for some other $r_A$. Otherwise, Bob will gain information about $x$ or $S_{[i]}$. Hence, there will always be a choice of $r_A$ which is consistent with the protocol.
This finishes the first step, and consequently, the whole proof.
\end{proof}

\section{Conclusion and open questions}
We have shown various upper and lower bounds on \nlb\ complexity, and shown how the upper bounds could be translated into bounds for secure function evaluation. We have also shown how to simulate quantum correlations arising from binary measurements on bipartite entangled states using 3 \nlbs. %\comment{Do not hesitate to add a little more details. For the next remark, I added it because it seemed interesting to me, I hope you agree ;-). Maybe this should be moved in the section on simulating quantum correlations, but the problem is that we have not introduced the simulation of NLB by OT at this point of the paper... } 
Note that combining these last two results also implies that such quantum correlations may also be simulated using 3 \OT\ boxes. The advantage is that, while \nlbs\ may not be actually realized (due to their violation of Tsirelson's bound), \OT\ boxes may be implemented under computational assumptions.
Note that such a simulation with \OT\ boxes breaks the timing properties of the EPR experiment, but this is unavoidable when simulating quantum correlations using classical resources, due to the violation of Bell inequalities. Moreover, contrary to a simulation with communication such as in~\cite{rt07}, using \OT\ boxes preserves the cryptographic properties of the experiment, that is, Alice does not learn anything about Bob's measurement (and vice versa).

During our investigations, we have come across a series of interesting open questions.
\begin{itemize}
 \item For randomized \nlb\ complexity in parallel, can we remove the XOR restriction, that is, is the case that  $\RNLp(f)\approx\RNLpx(f)= \epsrk(M_f)$? The proof for the deterministic does not carry over because of the inherent randomness of the \nlbs, which could be used to save on the number of \nlbs\ when some error probability is authorized.
 \item While the Disjointness function provides an example of exponential gap between parallel and general deterministic \nlb\ complexity, the gap disappears in the randomized model. Is it always the case that parallel and general randomized \nlb\ complexities are polynomially related?
 \item In general, \nlbs\ could be used in a different order on Alice and Bob's side. Does this provide any advantage, that is, are there functions for which $\RNL(f)<\RNLo(f)$?
 \item As for secure function evaluation, we proved that the communication complexity is a lower bound on \OT\ complexity only under some optimality assumption. Can this assumption be made without loss of generality?
 \item Finally, another interesting question is whether we can prove an analogue of Theorem~\ref{thm:ot-lower-bound} for \nlbs. Ideally, we would like to prove that for secure computation with ordered \nlbs, communication does not help. Indeed, due to the reduction from ordered \nlb\  protocols to \OT\ protocols and vice versa, this would imply that $\RNLo(f)$ is exactly $\OT_{\eps}(f)$, and not just an upper bound. This would of course provide even more motivation to study \nlb\ complexity in the context of secure function evaluation. Note that working with \nlbs\ instead of \OT\ boxes provides a few advantages. First, protocols using \nlbs\ (and no communication) are necessarily secure, even in the malicious model. Second, contrary to \OT\ protocols, such \nlb\ protocols do not require private randomness (except the inherent randomness of the \nlbs) to ensure security in this model. %Finally, the fact that \nlbs\ are symmetric could be another advantage.
\end{itemize}

\subsection*{Acknowledgements}
We would like to thank Troy Lee and Falk Unger for pointing out the 
deterministic protocol for Disjointness. We thank Ronald de Wolf 
for suggesting a randomized protocol for disjointness 
with bias $1/\log(n)$, using Valiant-Vazirani.

\bibliographystyle{alpha}
\bibliography{biblio}

\end{document}